\documentclass[letterpaper, 10 pt, conference]{ieeeconf} 
\IEEEoverridecommandlockouts
\usepackage{cite}
\usepackage{amsmath,amssymb,amsfonts}
\usepackage{algorithm}
\usepackage{algorithmic}
\usepackage{graphicx}
\usepackage{textcomp}
\usepackage{xcolor}
\usepackage{booktabs} 
\usepackage{siunitx} 
\usepackage{multirow}
\usepackage{subcaption}
\usepackage{caption}
\usepackage{threeparttable}
\usepackage{adjustbox}
\usepackage{siunitx}
\usepackage{booktabs}
\def\BibTeX{{\rm B\kern-.05em{\sc i\kern-.025em b}\kern-.08em
    T\kern-.1667em\lower.7ex\hbox{E}\kern-.125emX}}

\usepackage{microtype}
\usepackage{graphicx}
\usepackage{subcaption}
\usepackage{booktabs} 

\usepackage{hyperref}
\usepackage{multirow}

\usepackage{array}

\usepackage{amsmath}
\usepackage{amssymb}
\usepackage{mathtools}
\usepackage{amsthm}
\usepackage{subcaption}

\usepackage[capitalize,noabbrev]{cleveref}

\theoremstyle{plain}
\newtheorem{theorem}{Theorem}[section]

\newtheorem{lemma}[theorem]{Lemma}

\theoremstyle{definition}

\theoremstyle{remark}

\usepackage[textsize=tiny]{todonotes}

\title{\LARGE \bf
Preparation of Papers for IEEE CSS Sponsored Conferences \& Symposia
}

\author{Chengyang Gu, Yuxin Pan, Hui Xiong, and Yize Chen
\thanks{C. Gu and H. Xiong are with Information Hub, HKUST (Guangzhou), Guangzhou, China
        {\tt\small cgu893@connect.hkust-gz.edu.cn, xionghui@hkust-gz.edu.cn }}%
\thanks{Y. Pan is with the Department of Computer Science, City University of Hong Kong, Hong Kong, China
        {\tt\small yuxin.pan@connect.ust.hk}}%
\thanks{Y. Chen is with the Department of Electrical and Computer Engineering, University of Alberta, Edmonton, Canada
        {\tt\small yize.chen@ualberta.ca}}
}

\begin{document}

\title{A Pontryagin Method of Model-based \\Reinforcement Learning via Hamiltonian Actor-Critic}

\maketitle

\begin{abstract}
 Model-based reinforcement learning (MBRL) improves sample efficiency by leveraging learned dynamics models for policy optimization. However, effectiveness of methods such as actor–critic is often limited by compounding model errors that degrade long-horizon value estimation. Existing approaches such as Model-Based Value Expansion (MVE) partially mitigate this issue through multi-step rollouts, but remain sensitive to rollout horizon selection and residual model bias. Motivated by Pontryagin Maximum Principles (PMP), we propose Hamiltonian Actor-Critic (HAC), a model-based approach that eliminates explicit value function learning by directly optimizing a Hamiltonian defined over the learned dynamics and reward for systems with deterministic dynamics. By avoiding value approximation, HAC reduces sensitivity to model errors while admitting convergence guarantees. Extensive experiments on continuous control benchmarks across online and offline RL show that HAC outperforms model-free and MVE-based baselines in control performance, convergence speed, and robustness to distributional shift—including out-of-distribution (OOD) issues. In offline settings with limited data, HAC matches or exceeds state-of-the-art methods, showing strong sample efficiency.
\end{abstract}

\section{Introduction}

Although deep reinforcement learning (DRL) has demonstrated impressive empirical performance \cite{mnih2013playing,lillicrap2015continuous}, fundamental issues related to sample efficiency and generalization continue to pose challenges \cite{yu2018towards}. Model-based reinforcement learning (MBRL) provides one promising perspective to address such issues \cite{sutton1998reinforcement}. By learning an internal dynamics model or representations and training policy agents upon them, the number of required agent-environment interactions are significantly reduced \cite{plaat2023high}.

To ground MBRL for practical decision-making problems, learning accurate dynamics and value functions efficiently emerge as the fundamental research goals. The performance of MBRL is highly contingent on the quality of the model learning component,  particularly in Dyna-style frameworks \cite{sutton1991dyna}, which work by augmenting the replay buffer with imaginary rollouts generated by the learned model \cite{frauenknecht2024trust}. However, learning an accurate dynamics model in complex environments can be challenging. Model errors tend to compound across imagined state transitions \cite{xiao2019learning}, leading to inaccurate value estimates for state-action pairs \cite{nauman2024overestimation}. This problem becomes even more pronounced when using actor-critic architectures \cite{konda1999actor}, where the critic must learn from these imaginary rollouts. Misaligned value estimations can further mislead the policy agent, ultimately degrading control performance \cite{lin2023model}.

To improve quality of target values for critic learning, one effective strategy to mitigate the impact of model and value estimation errors is to to generate multi-step returns via learned models. By leveraging a multi-step lookahead, target critic values can incorporate more long-term information, leading to more accurate and comprehensive value estimates. Classical works using this technique, such as Model-Based Value Expansion (MVE) \cite{feinberg2018model} and its extensions STEVE \cite{buckman2018sample}, AdaMVE \cite{xiao2019learning}, MPPVE \cite{lin2023model} have demonstrated that integrating value expansion into actor-critic algorithms such as DDPG results in significant performance improvements. However, due to the inherent imperfection of learned dynamics models, MVE-based approaches remain vulnerable to compounding errors in the imagined rollouts. These errors can misguide the critic network, making its learning process inaccurate. As a result, the effectiveness of value expansion is highly sensitive to the choice of the \emph{rollout horizon} \cite{wang2020dynamic, georgiev2024adaptive}: longer horizons may introduce greater imaginary rollout bias, produce erroneous target values, and ultimately lead to larger value estimation error for critic learning \cite{palenicek2023diminishing, buckman2018sample}. This further undermines training stability and degrade the agent's overall performance.

While MVE-based methods struggle with cumulative errors in biased imaginary rollouts and the following inaccurate critic learning, recent insights from optimal control theory, a field closely related to MBRL, motivates us to explore a novel paradigm for policy optimization. In particular, the Pontryagin Maximum Principles (PMP) \cite{kopp1962pontryagin, mehta2009q}, which provide necessary conditions for optimal control over a fixed horizon, offer a more principled alternative: optimizing policies directly over trajectories, thereby avoiding the need for learning an explicit value function \cite{jin2020pontryagin, eberhard2024pontryagin}. This raises a compelling question: \textit{What if we also do not need a separate critic approximator in actor-critic MBRL}? Motivated by this idea, we introduce the \textbf{Hamiltonian Actor-Critic (HAC)} algorithm for both online and offline RL, under \textbf{systems with determinisitic dynamics}. During policy learning, instead of relying on potentially inaccurate critic value estimates, HAC directly minimizes the Hamiltonian — a quantity with theoretical ties to optimality. The Hamiltonian is analytically defined once the system dynamics and reward function are parameterized. We also design an analytical Jaco-
bian technique that for HAC, which enables graph-free
and efficient costate computation. This allows HAC to eliminate the need for an explicit critic network during training, yielding significant computational savings while preserving theoretical convergence guarantees. Moreover, we prove that under certain conditions, HAC ensures lower maximum value estimation error than MVE.

We empirically evaluate the proposed Hamiltonian Actor-Critic (HAC) on a diverse set of continuous control benchmarks with increasing complexity, starting from fundamental \textit{Linear Quadratic Regulator (LQR)} task, to \textit{Pendulum}, \textit{MountainCar}, and MuJoCo environments such as \textit{Swimmer} and \textit{Hopper}, under both online and offline reinforcement learning settings. Compared with standard model-free baselines (DDPG, SAC), their model-based extension MVE-DDPG, and state-of-the-art offline RL methods like IQL and MOPO, HAC consistently achieves superior or highly competitive control performance while converging faster, indicating improved sample efficiency. Moreover, HAC exhibits enhanced robustness to out-of-distribution conditions, such as shifted initial states, where MVE-based methods degrade. These empirical results corroborate our theoretical analysis, demonstrating tighter critic estimation error bounds translate into improved performance, sample efficiency, and robustness~\cite{barkley2024stealing}.


\section{Related Works}

\paragraph{Value Estimation in MBRL} Beyond the model-based value expansion (MVE) approach, a range of research efforts have focused on addressing value estimation errors in MBRL. These methods can be broadly categorized into two main strategies. The first involves using ensembles of critic networks \cite{huang2017learning, park2023model, an2021uncertainty,  lyu2022efficient}, where multiple critics are trained in parallel. To improve value estimates, agents either average the outputs of the ensemble to reduce variance \cite{chen2021randomized}, or adopt the minimum value across critics to avoid overestimation \cite{haarnoja2018soft, fujimoto2018addressing}. Despite achieving outstanding results, training an ensemble of models requires significant additional computational resources and time, rendering the method inefficient \cite{nikulin2022q, zheng2023model, duan2021distributional}. Another line of work addresses value estimation errors by penalizing critic values for out-of-distribution (OOD) states based on model uncertainty in state predictions \cite{yu2020mopo, sun2023model, cetin2023learning, luis2023model, wu2022plan}. This is typically achieved through constraints \cite{jeong2022conservative, kumar2019stabilizing} or penalty terms \cite{chen2023conservative, kumar2020conservative} that encourage conservative learning and discourage agents from venturing into poorly modeled regions. While mitigating the impact of overestimated or underestimated state-action values, the heuristic nature of penalties can distort the critic’s value distribution from ground-truth, often resulting in overly conservative or suboptimal policies \cite{parkmodel}.

\paragraph{Sample Efficient Reinforcement Learning} Several lines of research have also focused on improving sample efficiency in reinforcement learning through alternative approaches besides MVE. One direction involves enhancing policy gradient methods to accelerate learning and reduce variance in gradient estimates \cite{zhang2021sample, agarwal2021theory}. Another promising avenue involves simplifying function approximators, including dynamics models and critics. Recent studies have shown that using simple linear approximations \cite{jin2023provably, ding2021provably} or sparse neural networks \cite{liu2023model, hiraoka2021dropout} can substantially enhance both sample and computational efficiency with minimal losses in performance. These findings suggest a potential model surplus in current actor-critic architectures, where overly complex models may not be necessary for effective learning.

\paragraph{Learning Optimal Policies} Recent advances in the optimization field have leveraged deep learning to solve optimization problems by embedding them directly into neural network architectures \cite{amos2017optnet}. These approaches treat optimization problems as differentiable layers, allowing their objectives to be incorporated as loss functions and enabling end-to-end learning \cite{kong2022end, de2018end, donti2017task, chenreddy2024end}. This paradigm has been successfully extended to optimal control problems through the use of Pontryagin’s Maximum Principles. Notable examples include Pontryagin Differential Programming (PDP) \cite{jin2020pontryagin} and AI Pontryagin \cite{bottcher2022ai}, which demonstrate strong performance in structured decision-making tasks. However, these methods rely entirely on offline datasets and thus operate under a different setting from our work, which considers interactive Dyna-style model-based reinforcement learning.

\section{Preliminary}
\label{Pre}
\noindent\paragraph{Reinforcement Learning} We consider a deterministic Markov Decision Process (MDP) \cite{sutton1998reinforcement, puterman2014markov} which is defined by a tuple $\mathcal{M} = (\mathcal{S}, \mathcal{A}, f, r, \gamma)$, where $\mathcal{S} \subseteq \mathbb{R}^m$ and $\mathcal{A} \subseteq \mathbb{R}^n$ denote continuous state and action space respectively~\footnote{See our method's stochastic extension in Appendix \ref{app::stochastic}.}. $f: \mathcal{S} \times \mathcal{A} \rightarrow \mathcal{S}$ denotes the deterministic transition dynamics of the system. $r: \mathcal{S} \times \mathcal{A} \rightarrow \mathbb{R}$ is the reward function, and $\gamma \in (0, 1]$ is the discount factor that balances immediate and future rewards. $r_T$ denotes the terminal reward function for final state $s_T$. We assume we known $r$ and $r_T$. The RL objective is to learn a policy $\pi_{\theta}(a|s)$ parameterized by $\theta$, to maximize the cumulative reward alongside a discrete-time horizon with finite maximum timestep $T$:
\begin{align}
\max_{\theta} \quad & J(\theta) = \sum_{t=0}^{T-1} \gamma^t r(s_t, a_t) + r_T(s_T)\nonumber \\
\text{s.t.} \quad & s_{t+1} = f(s_t, a_t), \quad a_t = \pi_{\theta}(s_t) .
\label{eq:RL}
\end{align}

\noindent \paragraph{Actor-Critic} In this work, we focus on the inaccurate value estimation issue in Actor-Critic \cite{konda1999actor, silver2014deterministic} RL algorithms such as Deep Deterministic Policy Gradients (DDPG) \cite{lillicrap2015continuous} and Soft Actor-Critic (SAC) \cite{haarnoja2018soft}. In the actor-critic setting, with Q-value defined as $Q^\pi (s,a) = r(s,a) +\gamma V^\pi (s')$, not only the policy agent $\pi$, but the critic value function (i.e. Q-value) $Q_{\phi}(s_t, a_t)$ is also represented by parameterized function approximators such as neural networks. The critic network parameter $\phi$ is updated by minimizing the temporal difference:
\begin{equation}
    \min_{\phi} \quad  J(\phi) = (Q_{\phi}(s_t, a_t) - r(s_t, a_t) - \gamma Q_{\phi_{'}}(s_{t+1}, a_{t+1}))^2;
    \label{eq:standardQ}
\end{equation}
where $Q_{\phi_{'}}$ is the target network. Based on the learned critic network, the trainable agent $\pi_{\theta}$
performs the policy gradient update to maximize the estimated value of the chosen state-action pair:
\begin{equation}
    \nabla_{\theta} J(\theta) = \nabla_{\theta} \pi_{\theta}(s_t) Q_{\phi}(s_t, a_t) |_{a_t=\pi_{\theta}(s_t)}.
    \label{eq:policygradient}
\end{equation}
While this architecture enables efficient and stable learning, it often suffers from value overestimation issues \cite{fujimoto2018addressing}. In actor-critic methods, inaccurate or overly optimistic value estimates can compound over time due to bootstrapped updates from a learned critic function approximator. This issue becomes even more severe in Dyna-style model-based RL, where an approximated dynamics model $\hat{f}$ generates imaginary trajectories. If $\hat{f}$ is flawed, it produces unrealistic state transitions and incorrect reward predictions, further biasing the value targets. As a result, the actor is misled into exploiting overestimated returns, which degrades policy quality. This creates a self-reinforcing feedback loop: the actor explores more unrealistic regions, the model performs worse in those regions, the critic becomes more biased, and the actor continues to drift.



\paragraph{Model-based Value Expansion (MVE)} In this work, we compare our method with Model-based Value Expansion (MVE), which focuses on computing multi-step critic targets for critic network training using model-based imaginary rollouts \cite{palenicek2023diminishing, buckman2018sample, lin2023model}. In MVE, to improve long-term value estimation, the standard one-step Bellman target in temporal-difference learning (Eq. \eqref{eq:standardQ}) can be replaced by a $K$-step \textbf{critic expansion} that leverages rollouts from the learned dynamics model \cite{feinberg2018model}. This approach computes a multi-step return as the training target for the critic, leading to more accurate value estimates:
\begin{equation*}
    \begin{aligned}
        \min_{\phi} \quad  J(\phi) = (&Q_{\phi}(s_t, a_t) - \sum_{t=1}^{K-1}\gamma^{t-1} 
    r(s_t, a_t) \\ 
    &- \gamma^{K} Q_{\phi_{'}}(s_{K}, a_{K}))^2;
    \end{aligned}
    \label{eq:CriticExpansion}
\end{equation*}
    which is served as an alternative target for critic value learning. Empirically, such improved targets have been proved to accelerate learning and enhance performance, particularly in tasks with dense rewards. However, since this approach still relies on both a trained critic network and a learned dynamics model, it remains sensitive to model error, which can accumulate over longer imaginary rollouts. Moreover, the choice of the rollout horizon $K$  becomes critical. Tasks with varying dynamics complexity require careful tuning of $K$ to balance bias and variance in the target estimates \cite{xiao2019learning}. To address this, subsequent works have proposed adaptive horizon selection based on model confidence, but these techniques introduce additional complexity and overhead into the training pipeline .



\paragraph{Pontryagin Maximum Principles (PMP)} Our HAC is grounded in \textit{Pontryagin Maximum Principles (PMP)}. Given access to a learned dynamics model, the reinforcement learning (RL) problem in Eq. \eqref{eq:RL} can be reformulated as an optimal control problem:
\begin{align}
\min_{\theta} \quad & J(\theta) = \sum_{t=0}^{T-1} c(s_t, a_t) + \Phi_T(s_T)\nonumber \\
\text{s.t.} \quad & s_{t+1} = f(s_t, a_t), \quad a_t = \pi_{\theta}(s_t) .
\label{eq:optimization}
\end{align}
where the \textit{running cost} is defined as $c(s_t, a_t) = -\gamma^{t} r(s_t, a_t)$, and the \textit{terminal cost} is defined by $\Phi_T(s_T) = - r_T(s_T)$. To solve this reformulated optimal control problem, the PMP introduces the \textit{Hamiltonian}, a new objective for minimization, which is defined as:
\begin{equation}
    H(s_t, a_t, \lambda_{t+1}) = c(s_t, a_t) + \lambda_{t+1}^{T} f(s_t, a_t).
    \label{eq:Hamiltonian}
\end{equation}
where $\lambda_{t}$ is defined as \textit{costate}, computed via backward recursion grounded on PMP conditions:
\begin{align}
   \lambda_t  &= \nabla_{s_t} H(s_t, a_t, \lambda_{t+1}) \\ \notag
   &= \nabla_{s_t} c(s_t, a_t) + \nabla_{s_t}f(s_t, a_t)\lambda_{t+1}; \\
 &\lambda_T = \nabla_{s_T} \Phi_T(s_T).
\label{eq:costate}
\end{align}
From Eq. \eqref{eq:Hamiltonian} and \eqref{eq:costate}, it naturally follows that the Hamiltonian $H(s_t, a_t, \lambda_{t+1})$ can serve as a \textbf{direct, equivalent replacement for critics} $Q_{\phi}(s_t, a_t)$ in Eq. \eqref{eq:policygradient}, provided the dynamics model $f$ is known or learned. This formulation eliminates the need to train a separate critic network to approximate value functions, thereby simplifying the architecture and mitigating the progressive build-up of errors commonly seen in model-based actor-critic frameworks.

\section{Hamiltonian Actor-Critic Approach}
\label{Method}
Building on the insight that the Hamiltonian from Pontryagin’s Maximum Principles (PMP) can serve as a dynamics-aware surrogate for the critic in MBRL, we propose a novel method called \textbf{Hamiltonian Actor-Critic (HAC)}. HAC replaces the separate critic networks in traditional actor-critic frameworks with a Hamiltonian-based formulation, integrated with $K$-step imaginary rollouts. We begin by detailing the HAC algorithm, followed by an analysis of its convergence properties. Finally, we provide the theoretical guarantee that under certain conditions, HAC yields a smaller maximum critic estimation error compared to the classical multi-step critic expansion method, Model-based Value Expansion (MVE).

\subsection{HAC Training}
Here we present the detailed implementation of the HAC algorithm, demonstrating how the Hamiltonian rooted in PMP can be seamlessly incorporated into the model-based actor-critic framework. 

Our HAC is built upon a parameterized policy agent $\pi_{\theta}(s_t)$, and a pre-trained dynamics model $\hat{f}$. Recall from Section \ref{Pre} that our goal is to compute the Hamiltonians along a 
$K$-step ahead imaginary rollout, and use them as negative critic values to minimize for each state-action pair. To generate these rollouts, we rely on the current policy agent $\pi_{\theta}$ and the approximated dynamics $\hat{f}$. Specifically, the imaginary rollout is constructed as follows: $\hat{s}_{t+1}=\hat{f}(\hat{s}_{t}, \hat{a}_t)$, where $\{\hat{s}_{1:K} \}$ yielding the sequence of simulated states and $\hat{a}_t = \pi_{\theta}(\hat{s}_t)$ are actions based on the policy. Given these rollouts, we compute the corresponding costates $\hat{\lambda}_{1:K-1}$ using the backward recursion derived from Eq. \eqref{eq:costate}: $\hat{\lambda}_t=-\nabla_{\hat{s}_t} \gamma^{t} r(\hat{s}_t, \hat{a}_t) + \nabla_{\hat{s}_t}\hat{f}(\hat{s}_t, \hat{a}_t)\hat{\lambda}_{t+1}, \quad \text{for} \quad t=K-1, ..., 1$, with the terminal costate $\hat{\lambda}_K$ initialized as: $\hat{\lambda}_K =
\begin{cases}
\nabla_{\hat{s}_K} r_T(\hat{s}_K), & \text{if } r_T \text{ is defined}, \\
0, & \text{otherwise}.
\end{cases}$

Using the computed simulated costates and the known reward functions $r$ and $r_T$, we establish Hamiltonian for $H(\hat{s}_{t}, \hat{a}_t, \hat{\lambda}_{t+1})$ for each state-action pair $(\hat{s}_{t}, \hat{a}_t)$, treating it as a negative critic value. Our objective is to adjust the policy agent so that the selected action $\hat{a}_t$ minimizes the corresponding Hamiltonian. Accordingly, the actor network's loss is defined as:
\begin{equation}
    L_{actor}(\theta) = \frac{1}{K}\sum_{t=0}^{K-1}  H(\hat{s}_{t}, \hat{a}_t, \hat{\lambda}_{t+1}), \quad \hat{a}_t = \pi_{\theta}(\hat{s}_t).
    \label{eq: actorloss}
\end{equation}

Following the standard actor-critic settings, we incorporate a target actor network $\pi_{\theta'}$ and update the primary actor $\pi_\theta$ using a soft update strategy with a decay factor. The overall procedure of the HAC algorithm is outlined in Algorithm \ref{alg:hac}.

One interesting fact is HAC framework is applicable to both online and offline RL settings. Unlike conventional actor–critic methods, which often struggle in offline RL due to unreliable critic estimates for unseen actions outside the limited collected trials, HAC eliminates the need for explicit critic learning. Instead, the optimization objective in Eq.~\eqref{eq: actorloss} is solely defined  by the pre-trained dynamics model, making HAC less sensitive to distributional shift in offline data. 

\begin{algorithm}[tb]
  \caption{Hamiltonian Actor-Critic (HAC)}
  \label{alg:hac}
  \begin{algorithmic}
    \STATE {\bfseries Input:} Parameterized policy $\pi_\theta$, target policy $\pi_{\theta'}$, pre-trained dynamics model $\hat{f}$, reward functions $r$, $r_T$, rollout horizon $K$, soft update factor $\tau$, max iterations number $N$, learning rate $\eta$, replay buffer $D$.
    \REPEAT
    \STATE (Online) Collect transitions $(s_t, a_t, s_{t+1})$ with $\pi_{\theta}$, add to replay buffer $D$.
    \STATE Fit dynamics $\hat{f}$:  \\$\hat{f} = \text{argmin}_{(s_t, a_t, s_{t+1}) \in D} || \hat{f}(s_t, a_t) - s_{t+1}||^2$.
    \FOR{$\text{iters}=1$ to $N$}
    \STATE Sample initial state $\hat{s}_0$
    \FOR{$t = 0$ to $K-1$}
        \STATE $\hat{a}_t \gets \pi_\theta(\hat{s}_t)$; $\hat{s}_{t+1} \gets \hat{f}(\hat{s}_t, \hat{a}_t)$
    \ENDFOR
    \IF{$r_T$ is defined}
    \STATE $\hat{\lambda}_K \gets \nabla_{\hat{s}_K} r_T(\hat{s}_K)$
    \ELSE
    \STATE $\hat{\lambda}_K \gets 0$
    \ENDIF
    \FOR{$t = K-1$ to $1$}
        \STATE $\hat{a}_t \gets \pi_\theta(\hat{s}_t)$
        \STATE $\hat{\lambda}_t \gets -\nabla_{\hat{s}_t} \gamma^t r(\hat{s}_t, \hat{a}_t) + \nabla_{\hat{s}_t} \hat{f}(\hat{s}_t, \hat{a}_t) \hat{\lambda}_{t+1}$
    \ENDFOR

    \STATE Compute Hamiltonian: $H(\hat{s}_t, \hat{a}_t, \hat{\lambda}_{t+1})$ for $\forall t$ 
    \STATE Compute actor loss:\\  $L_{actor}(\theta) = \frac{1}{K} \sum_{t=0}^{K-1} H(\hat{s}_t, \hat{a}_t, \hat{\lambda}_{t+1})$
    \STATE Update $\theta$ using gradient descent on $L_{actor}$:\\ $\theta \leftarrow \theta - \eta\nabla_\theta{L_{actor}(\theta)}$
    \STATE Soft update target policy: $\theta' \gets \tau \theta + (1 - \tau) \theta'$
    \ENDFOR
    \UNTIL{Converged}
  \end{algorithmic}
\end{algorithm}

\subsection{Technique: Analytical Jacobian}
In the Hamiltonian Actor--Critic (HAC) algorithm, the policy is updated by optimizing the estimated Hamiltonian $H(\hat{s}_t, \hat{a}_t, \hat{\lambda}_{t+1})$ over a $K$-step rollout. Computation of the costate requires the Jacobian of the pre-trained dynamics, 
\begin{equation}
\hat{\lambda}_t
=
- \nabla_{\hat{s}_t} \gamma^t r(\hat{s}_t, \hat{a}_t)
+
\left( \nabla_{\hat{s}_t} \hat{f}(\hat{s}_t, \hat{a}_t) \right)^\top
\hat{\lambda}_{t+1}.
\end{equation}
However, directly computing $\nabla_{\hat{s}_t} \hat{f}$ using automatic differentiation in PyTorch presents two major challenges: First, the Hamiltonian gradient requires reuse of the computational graph, resulting in RunTime Errors as costate computation back-propagates through the same graph. Second, the costate recursion demands Jacobian at every time step, and autograd-based computation incurs repeated higher-order gradient operations, substantially increasing computational cost and slowing down training, in either offline or online modes. Such an issue is also encountered in several learning to optimize and implicit differentiation tasks~\cite{agrawal2019differentiable, bolte2021nonsmooth}.

\vspace{-5pt}To address these issues, we develop an \textbf{Analytical Jacobian Technique} that for HAC, which enables graph-free and efficient costate computation leveraging the preserved parameters of the pre-trained dynamics model.

Throughout this paper, we model the dynamics using a neural network with ReLU activations, which is shown to be effective in our simulations: 
\begin{align*}\vspace{-15pt}
z_t &= \begin{bmatrix} s_t , a_t \end{bmatrix}^T, \\
h_1 &= W_1 z_t + b_1, \quad c_1 = \sigma(h_1), \\
h_2 &= W_2 c_1 + b_2, \quad c_2 = \sigma(h_2), \\
s_{t+1} &= \hat{f}_\theta(s_t, a_t) = W_3 c_2 + b_3,
\end{align*}
where $\sigma(\cdot) = \mathrm{ReLU}(\cdot)$ and $W_1^s$ denotes the submatrix of $W_1$ corresponding to the state input.

Since $\sigma'(h) = \mathbb{I}[h > 0]$, we can define and identify diagonal activation matrices
\begin{equation}
D_1 = \mathrm{diag}\big( \mathbb{I}[h_1 > 0] \big), \qquad
D_2 = \mathrm{diag}\big( \mathbb{I}[h_2 > 0] \big).
\end{equation}
By applying the chain rule, the Jacobian of the neural dynamics with respect to the state admits the following closed-form expression
\begin{equation}
\label{eq:analytical_jacobian}
\nabla_{s_t} \hat{f}(s_t, a_t) 
=
W_3 D_2 W_2 D_1 W_1^s.
\end{equation}
For batched inputs, the activation matrices are computed independently for each sample, enabling efficient GPU implementation using only matrix multiplications and element-wise masking.

\subsection{Convergence of HAC Learning}
To ensure the HAC scheme could improve policy learning, we show that by iteratively updating the policy agent $\pi_\theta$ using the gradient $\nabla_\theta{L_{actor}(\theta)}$, the HAC algorithm is guaranteed to converge to a local optimum.

\begin{lemma}
\label{lemma::1}
    In the actor network learning objective Eq. \eqref{eq:optimization}, updating $a_t$ with $a_t \leftarrow a_t - \eta \nabla_{a_t} H(s_t, a_t, \lambda_{t+1})$ is equivalent to taking gradient descent of $J = \sum_{t=0}^{T-1} c(s_t, a_t) + \Phi_T(s_T)$ with respect to $a_t$ : $a_t \leftarrow a_t - \eta \nabla_{a_t} J$. 
\end{lemma}  

See Proof in Appendix \ref{proofL1}.  Since $a_t = \pi_{\theta}(s_t)$, with \textit{Lemma} \ref{lemma::1} and the chain rule, we can show the following results on the  convergence of HAC:

\begin{theorem}
\label{theorem::Convergence}
    (Convergence of HAC) Suppose the following hold: (i) $J$ is differentiable  and Lipschitz continuous with constant $L_m$ w.r.t. $a_t$; (ii) $\nabla_{a_t} J$ is Lipschitz continuous with contant $L_J$; (iii) $\pi_{\theta}$ is differentiable  and Lipschitz continuous with constant $L_G$ w.r.t. $\theta$; (iv) $\nabla_{\theta} \pi_{\theta}$ is $L_p$-Lipschitz continuous. The policy agent updating $\theta \leftarrow \theta - \eta L_{actor}(\theta)$ in HAC guarantees convergence, with learning rate $\eta < 2/(L_J L_G^2+L_pL_m)$.
\end{theorem}

\textit{Proof:} See Proof in Appendix \ref{proofT1}. 

Since neural networks with standard activations (e.g. ReLU, Sigmoid, etc.) are commonly used to parameterize for both $\pi_{\theta}$ and $\hat{f}$, the assumptions in \textit{Theorem}~\ref{theorem::Convergence} are mild and readily satisfied under standard conditions.

\subsection{Tighter Bound of Critic Estimation Error}
In this section, we further analyze the conditions under which our HAC algorithm achieves a smaller maximum critic estimation error compared to MVE, or analytical tractability, we focus on widely used assumptions: the step reward $r(s_t, a_t)$ is quadratic, and terminal reward $r_T$ is either constant, linear or quadratic. Let $f$ denote the ground-truth dynamics. 

We define the critic estimation error in HAC as the difference between the estimated Hamiltonian $H(\hat{s}_t, \hat{a}_t, \hat{\lambda}_{t+1})$ and the ground-truth Hamiltonian $H_g(s_t,a_t, \lambda_{t+1})$: $|H-H_g|$, where $H$ is computed from $K$-step rollouts using the learned dynamics $\hat{f}$, and $H_g$ is derived from rollouts under the ground-truth dynamics $f$. Similarly, we define the critic value estimation error in MVE as the difference between the estimated Q value $Q_{\phi}(s_t, a_t)$ and the ground-truth Q value $Q(s_t, a_t)$: $| Q_{\phi}(s_t, a_t) - Q(s_t, a_t) |$, both of which are trained using $K$-step target values based on $\hat{f}$ and $f$ respectively, as described in Eq. \eqref{eq:CriticExpansion}.

To compare the critic estimation errors of HAC and MVE, we first address the main challenge: analyzing the Hamiltonian error in HAC. This difficulty arises from the recursive nature of the costate computation in Eq. \eqref{eq:costate}, which makes it challenging to accurately approximate the Hamiltonian at each time step. Drawing inspiration from Hamiltonian dynamics in mechanical systems, we observe that when the Hamiltonian lacks explicit time dependence, system with a symplectic structure keeps the Hamiltonian constant over time \cite{greydanus2019hamiltonian}. Notably, this symplectic structure is also present in optimal control problems, when Pontryagin’s Maximum Principle (PMP) is applied to define the costates and Hamiltonian. Hence, we have a similar Hamiltonian conservation law in the continuous-time formulation of optimal control problem: 
\begin{lemma} 
\label{lemma:conservative}
    (Hamiltonian Conservation Law in Optimal Control) The Hamiltonian of the continuous-time optimal control problem:  $\int_{t=0}^T r(s(t), a(t))+r_T(s(T)), \text{s.t.}\;  \dot{s} = f(s(t), a(t))$, is constant across time when applying the assumed optimal $a(t)$ from PMP.
\end{lemma}
\begin{proof}
Our proof is motivated by \cite{bourdin2019continuity}. In continuous-time PMP, Hamiltonian is defined as $H=r(s(t), u(t))+\lambda(t)f(s(t), a(t))$. The derivative $\frac{\partial H}{\partial t} = \frac{\partial H}{\partial s} \frac{\partial s}{\partial t} + \frac{\partial H}{\partial a} \frac{\partial a}{\partial t} + \frac{\partial H}{\partial  \lambda} \frac{\partial \lambda}{\partial t}$. Easy to find that $\frac{\partial s}{\partial t} =  \frac{\partial H}{\partial \lambda}$. With costate backward recursive $-\frac{\partial \lambda}{\partial t} = \frac{\partial H}{\partial s}$ in PMP, $\frac{\partial H}{\partial s} \frac{\partial s}{\partial t} +\frac{\partial H}{\partial  \lambda} \frac{\partial \lambda}{\partial t} = 0$. Since $a(t)$ is assumed optimal, $\frac{\partial H}{\partial a} = 0$. Hence, $\frac{\partial H}{\partial t} = 0$. 
\end{proof}

The Hamiltonian Conservation Law stated in \textit{Lemma} \ref{lemma:conservative} also holds in the discrete-time optimal control setting Eq. \eqref{eq:optimization} when using midpoint discretization~\cite{bijalwan2023control}. Leveraging this conservation property, we avoid accumulating compounding errors across the rollout horizon during forward simulation and backward recursion. Instead, we can evaluate the Hamiltonian value at a single time step, specifically at $t=K-1$. This significantly simplifies the analysis and makes it easier to compute the estimation error $| H-H_g |$. Based on \textit{Lemma} \ref{lemma:conservative}, we have:

\begin{theorem}
    (Tighter Bound of Critic Estimation Error in HAC) When HAC and MVE converge to their optimal solutions, suppose: (i) the learned dynamics $\hat{f}$ in both HAC and MVE satisfy the generalization error $\delta$ in $K$-step depth (ii) the error $| Q_{\phi'} - Q|$ between target critic network $Q_{\phi'}(s_t, a_t)$ and ground-truth $Q(s_t, a_t)$ is bounded by $\epsilon$. If: (a) $r(s_t, a_t): = -(s_t^{T}Rs_t+a_t^TPa_t)$ is quadratic; (b) $r_T(s_K):= C$ (Constant),  $As_T+b$ (Linear), or $s_T^{T}As_T$ (Quadratic), we have: $\max 
|H-H_g | \leq \max | Q_{\phi} - Q| $.
\label{theorem::bound}
\end{theorem}

\textit{Proof:} See Proof in Appendix \ref{proofT2}. 

We note that our HAC is applied only to systems with deterministic dynamics; for potential stochastic extensions, see Appendix~\ref{app::stochastic}.


\section{Experiments}
\label{exp}

We evaluate HAC on a suite of continuous, deterministic control environments with increasing levels of complexity, including \textit{Linear Quadratic Regulator (LQR), Pendulum, MountainCar, MuJoCo Swimmer, and Hopper}, under both \textit{online} and \textit{offline} reinforcement learning settings. These environments are characterized by quadratic running costs and constant, linear, or quadratic terminal rewards, consistent with the assumptions of \textit{Theorem}~\ref{theorem::bound}. Among them, \textit{LQR} and \textit{Pendulum} represent short-horizon control tasks, while the remaining environments involve longer horizons and more complex dynamics.

Through these experiments, we empirically validate whether the tighter bound on critic estimation error established in Theorem~\ref{theorem::bound} improves our HAC's performance, and compare against classical actor–critic methods such as DDPG and SAC, as well as their Model-Based Value Expansion (MVE) extensions. Detailed experimental configurations are provided in Appendix~\ref{expset}.


Specifically, our experiments are designed to answer the following questions:
\begin{itemize}
\item \textbf{Control performance:} In online RL settings, does HAC achieve superior task performance compared to MVE-DDPG in an efficient way?
\item \textbf{Sample efficiency:} In offline RL with limited datasets, does HAC consistently identify generalizable actions which outperform state-of-the-art baselines?
\item \textbf{Robustness to OOD conditions:} Does HAC enhance robustness to out-of-distribution (OOD) scenarios, such as variations in initial states?
\end{itemize}

\subsection{Online RL Cases}

\begin{figure}[htbp]
    \centering

    \begin{subfigure}[t]{0.45\linewidth}
        \centering
        \includegraphics[width=\linewidth]{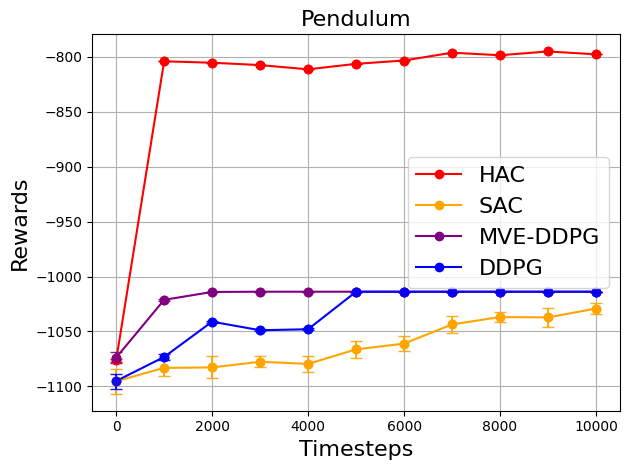}
        \caption{Pendulum}
    \end{subfigure}
    \hfill
    \begin{subfigure}[t]{0.45\linewidth}
        \centering
        \includegraphics[width=\linewidth]{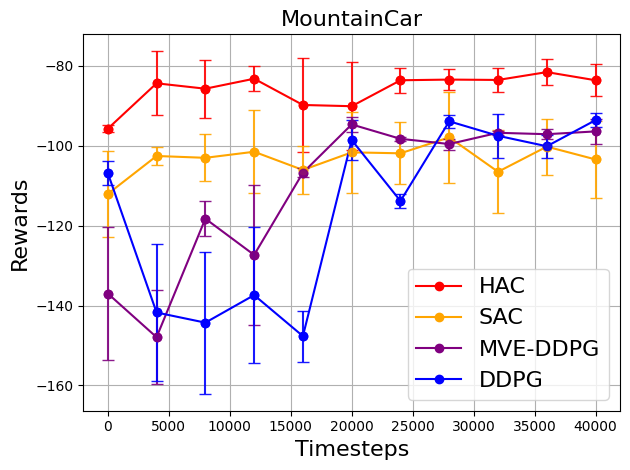}
        \caption{MountainCar}
    \end{subfigure}

    \vspace{0.6em}

    \begin{subfigure}[t]{0.45\linewidth}
        \centering
        \includegraphics[width=\linewidth]{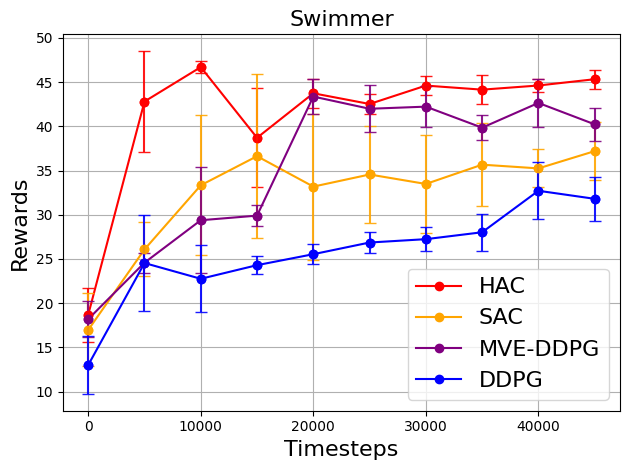}
        \caption{Swimmer}
    \end{subfigure}
    \hfill
    \begin{subfigure}[t]{0.45\linewidth}
        \centering
        \includegraphics[width=\linewidth]{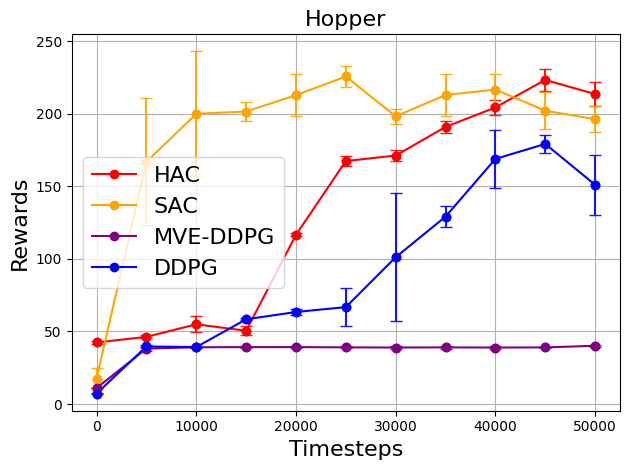}
        \caption{Hopper}
    \end{subfigure}
    \caption{\textbf{Learning Curves Comparison (Online RL):}  We compare our HAC with DDPG, SAC and MVE-DDPG on \textit{Pendulum}, \textit{MountainCar}, \textit{Swimmer} and \textit{Hopper}. For MVE-DDPG and HAC which require multi-step imaginary rollouts, we set the rollout horizon to $K=10$ for \textit{Pendulum} to cover all control horizon, $K=5$ for \textit{MountainCar} and \textit{Swimmer}, and $K=3$ for \textit{Hopper}. All results are averaged over 5 random seeds.}
    \label{fig:online_curves}
\end{figure}

We first evaluate the primary objective of reinforcement learning agents in terms of control performance. As established in \textit{Theorem}~\ref{theorem::bound}, HAC achieves a tighter upper bound on critic estimation error than MVE after convergence. This theoretical advantage suggests that HAC may yield more accurate value estimates, which in turn can lead to improved control performance and higher accumulated rewards. To empirically validate this hypothesis, we compare the proposed HAC method against representative baseline algorithms, including DDPG, SAC, and MVE-DDPG, across 4 benchmark environments (\textit{Pendulum}, \textit{MountainCar}, \textit{Swimmer} and \textit{Hopper}) in online reinforcement learning settings. For all tasks, the policy is updated every 100 environment timesteps. For the model-based approaches, HAC and MVE-DDPG, we employ the same pre-trained dynamics model as literature to ensure a fair comparison.

Fig.~\ref{fig:online_curves} presents the average learning curves (with standard deviations) over 5 random seeds for each method. As shown in the results, HAC significantly outperforms all baseline methods on the short-horizon \textit{Pendulum} task, thanks to the more accurate critic estimation (see in Appendix \ref{app::critic_error}). In longer-horizon environments, such as  \textit{MountainCar, Swimmer}, and \textit{Hopper}, HAC continues to achieve superior or highly competitive control performance compared to DDPG, SAC, and MVE-DDPG. Although our Hamiltonian estimation in HAC may exhibit deviations from the ground-truth due to the limited coverage of the control horizon in imaginary rollouts (i.e. $K<T$, see detailed analysis in Appendix~\ref{appendix::error_analyze}), these imperfections do not prevent HAC from attaining strong empirical performance. Overall, these results indicate that HAC converges to more effective control solutions and demonstrates greater adaptability across environments of varying complexity than classical actor–critic methods and their MVE-based extensions.

We observe that both HAC and MVE-DDPG encounter some challenges on the \textit{Hopper} task. Specifically, HAC exhibits slow convergence first, while MVE-DDPG becomes trapped in suboptimal solutions. This behavior can be attributed to the presence of an implicit \textit{healthy reward} mechanism in \textit{Hopper}, in addition to the standard quadratic cost. When the agent’s state violates the predefined health conditions, the episode terminates early at a time $t<T$. Under such circumstances, inaccuracies in the pre-trained dynamics model may cause the agent to enter unhealthy state regions during imaginary rollouts, leading to degraded performance for both methods.

\subsection{Offline RL Cases}

We also evaluate HAC under offline RL settings using limited pre-collected datasets, with the goal of assessing its ability to effectively exploit offline data and fastly learn policies. For each environment, the offline dataset is constructed from randomly collected transition trajectories (see Appendix \ref{expset}). To provide a comprehensive comparison, we benchmark HAC against widely used, state-of-the-art model-free offline RL baselines, including Implicit Q-Learning (IQL) \cite{kostrikov2021offline} and Offline SAC (SAC-Off) \cite{yu2023actor}, as well as a representative model-based offline RL method, Model-based Offline Policy Optimization (MOPO) \cite{yu2020mopo}.

\begin{table}[htbp]
\caption{Offline RL performance (rewards, averaged over 10 trials).}
\label{tab:offline_hac}
\centering
\footnotesize
\setlength{\tabcolsep}{3pt}
\begin{tabular}{m{1.6cm}cccc}
\toprule
Env & IQL & SAC-Off & MOPO & HAC \\
\midrule
LQR
 & \shortstack{-16.5\\{$\pm$0.2}}
 & \shortstack{-19.7\\{$\pm$1.3}}
 & \shortstack{-19.0\\{$\pm$1.1}}
 & \shortstack{\textbf{-15.0}\\{$\pm$0.0}} \\
\midrule
Pendulum
 & \shortstack{-999.8\\{$\pm$6.4}}
 & \shortstack{-1021.0\\{$\pm$9.6}}
 & \shortstack{-1017.1\\{$\pm$1.3}}
 & \shortstack{\textbf{-803.0}\\{$\pm$0.1}} \\
\midrule
MountainCar
 & \shortstack{-97.1\\{$\pm$1.4}}
 & \shortstack{-98.8\\{$\pm$0.7}}
 & \shortstack{-94.3\\{$\pm$0.3}}
 & \shortstack{\textbf{-83.5}\\{$\pm$10.9}}\\
\midrule
Swimmer
 & \shortstack{32.6\\{$\pm$11.1}}
 & \shortstack{31.7\\{$\pm$4.7}}
 & \shortstack{36.7\\{$\pm$8.3}}
 & \shortstack{\textbf{51.9}\\{$\pm$7.4}} \\
\midrule
Hopper
 & \shortstack{47.4\\{$\pm$0.9}}
 & \shortstack{160.9\\{$\pm$118.1}}
 & \shortstack{\textbf{210.3}\\{$\pm$202.9}}
 & \shortstack{167.8\\{$\pm$7.3}} \\
\bottomrule
\end{tabular}
\vspace{-0.12in}
\end{table}

Table~\ref{tab:offline_hac} reports the performance of HAC and all baseline methods in offline reinforcement learning settings across five environments. The results demonstrate that HAC achieves control performance that is competitive with its online counterpart, despite being trained on datasets of limited size. Moreover, HAC maintains a substantial performance advantage on short-horizon tasks such as \textit{LQR} and \textit{Pendulum}, while achieving competitive returns on long-horizon tasks relative to state-of-the-art offline RL baselines. This behavior is consistent with the online RL setting and can be attributed to the limited coverage of the control horizon in imaginary rollouts.

Additionally, as shown by the learning curves in Fig.~\ref{fig:offline_curves}, HAC consistently exhibits faster convergence than all baseline methods. Across both short-horizon and long-horizon tasks (except \textit{Hopper}, which is discussed in the online RL section), HAC converges to near-optimal solutions in substantially fewer training epochs, indicating reduced sample requirements  (about 25\% to 30\% samples compared with other baselines). These results prove that HAC achieves strong sample efficiency in offline reinforcement learning settings and is capable of solving challenging control problems using limited pre-collected datasets.

\begin{figure}[t]
    \centering

    \begin{subfigure}[t]{0.45\linewidth}
        \centering
        \includegraphics[width=\linewidth]{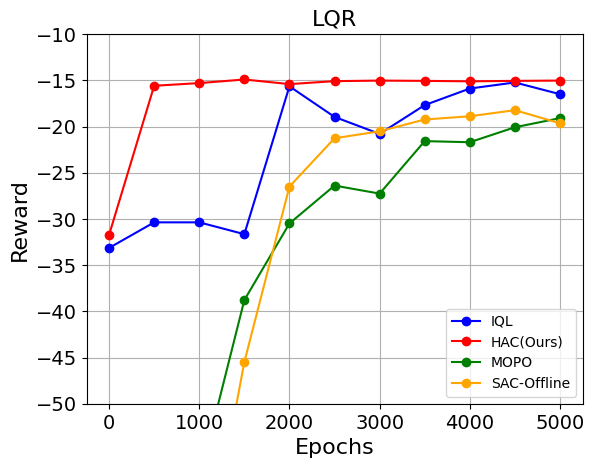}
        \caption{LQR}
    \end{subfigure}
    \hfill
    \begin{subfigure}[t]{0.45\linewidth}
        \centering
        \includegraphics[width=\linewidth]{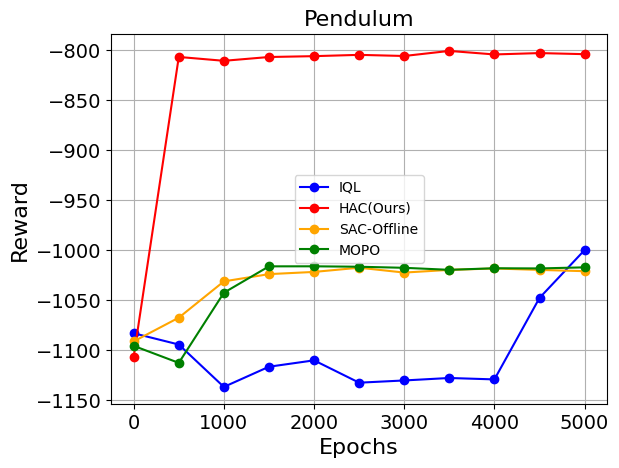}
        \caption{Pendulum}
    \end{subfigure}

    \vspace{0.6em}

    \begin{subfigure}[t]{0.45\linewidth}
        \centering
        \includegraphics[width=\linewidth]{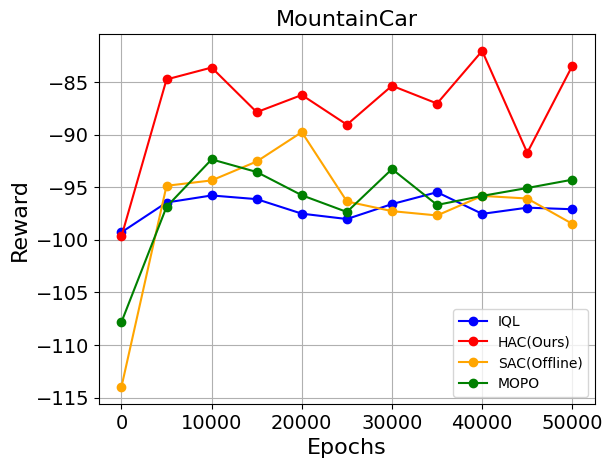}
        \caption{MountainCar}
    \end{subfigure}
    \hfill
    \begin{subfigure}[t]{0.45\linewidth}
        \centering
        \includegraphics[width=\linewidth]{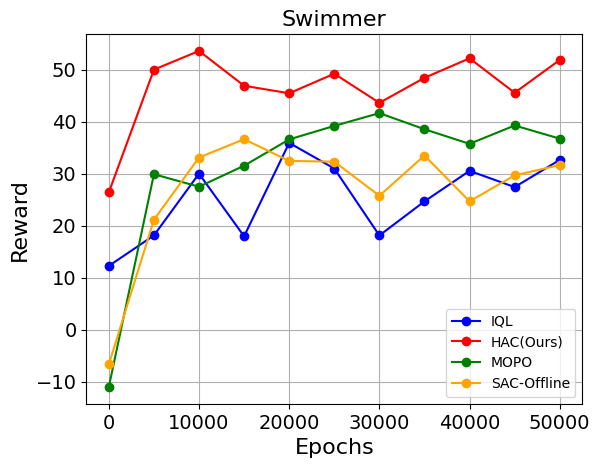}
        \caption{Swimmer}
    \end{subfigure}
    \caption{\textbf{Learning Curves Comparison (Offline RL)}. We compare our HAC with IQL, SAC-Off and MOPO on  \textit{LQR}, \textit{Pendulum}, and \textit{MountainCar}, \textit{Swimmer}. For MOPO and HAC which require multi-step imaginary rollouts, rollout horizon is to $K=10$ for \textit{LQR} and \textit{Pendulum} to cover all control horizon, $K=5$ for \textit{MountainCar} and \textit{Swimmer}. Results are averaged over 5 random seeds.}
    \label{fig:offline_curves}
\end{figure}

\subsection{Robustness and Generalization Capabilities}

\begin{figure}[h]
    \centering
    \begin{subfigure}[b]{0.5\linewidth}
        \centering
        \includegraphics[width=\linewidth]{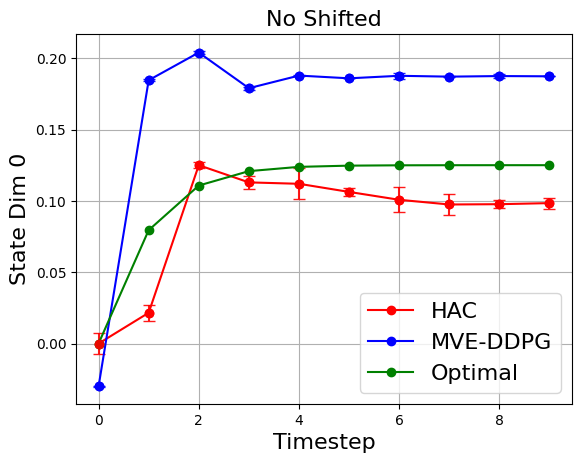}
        \caption{Trajectories with no initial state shift.}
        \label{fig:TrajNoShift}
    \end{subfigure}%
    \hfill
    \begin{subfigure}[b]{0.5\linewidth}
        \centering
        \includegraphics[width=\linewidth]{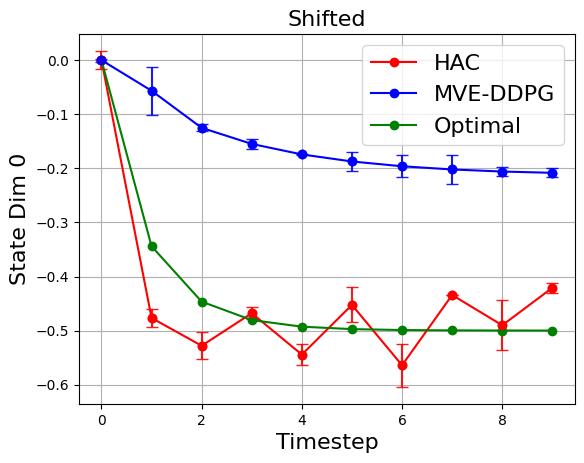}
        \caption{Trajectories with shifted initial state.}
        \label{fig:TrajShifted}
    \end{subfigure}
    \caption{\textbf{Robustness comparison between MVE-DDPG and proposed HAC against initial state shift}. The experiment is conducted on a 10-step LQR task with state dimension 5 and action dimension 3. Plots (a) and (b) illustrate optimal trajectories generated by HAC and MVE-DDPG in dimension 0 (averaged on 10 random seeds) under non-shifted and shifted initial states.}
    \label{fig:Trajs}
\end{figure} 
In this subsection, we further evaluate the robustness of our HAC method to out-of-distribution (OOD) issues commonly encountered in model-based reinforcement learning. When the critic network is trained on limited data, actions or states outside the training distribution can lead to significant estimation errors, resulting in poor and unstable control performance during policy execution. Since HAC demonstrates more accurate critic estimation, we investigate whether this accuracy also translates into improved robustness under OOD conditions.
Specifically, we introduce external disturbances by shifting and differing the initial state of the environment between training and testing. This causes the policy agent to encounter state-action pairs that lie outside the distribution seen during training, where the learned critic may fail to provide accurate value estimates. Under these conditions, we are interested in evaluating whether our HAC method can produce more consistent and reliable solutions.

We conduct this experiment in the 10-step LQR environment (See details in Appendix~\ref{expset}) . As shown in Fig.~\ref{fig:Trajs}, both methods perform well in the non-shifted case. However, in the shifted scenario, MVE-DDPG fails to produce stable trajectories, showing erratic behavior and significant divergence, especially in later timesteps. In contrast, our HAC maintains stable and consistent performance, demonstrating lower sensitivity to the disturbance. These results confirm that HAC is more robust to OOD issues compared to MVE. 

\section{Conclusion}
\label{con}
In this paper, we propose Hamiltonian Actor-Critic (HAC), a novel, efficient MBRL algorithm that replaces conventional value-function critics with a Hamiltonian-based surrogate derived from Pontryagin Maximum Principles (PMP), for systems with deterministic dynamics. We establish convergence guarantees for HAC under mild smoothness and Lipschitz continuity assumptions, and further derive a tighter upper bound on critic estimation error compared to Model-Based Value Expansion (MVE). Extensive experiments  empirically validate these theoretical advantages. In online RL, HAC consistently achieves superior or highly competitive control performance, while demonstrating faster convergence. In offline RL with limited datasets, HAC outperforms state-of-the-art model-free baselines and remains competitive with model-based methods. HAC is also robust against initial state shifts.  Together, these results confirm that tighter critic estimation error bounds translate into improved performance, efficiency, and robustness, establishing HAC as an effective and principled alternative to standard actor–critic frameworks.

\bibliographystyle{IEEEtran}
\bibliography{ref}

\clearpage
\onecolumn
\appendix
\section{Theoretical Proofs}
\label{proofs}

\subsection{Proof of Lemma 4.1}
\label{proofL1}
\paragraph{Lemma 4.1.} \textit{In optimal control formulation Eq. \eqref{eq:optimization}, updating $a_t$ with $a_t \leftarrow a_t - \eta \nabla_{a_t} H(s_t, a_t, \lambda_{t+1})$ is equivalent to taking gradient descent of $J = \sum_{t=0}^{T-1} c(s_t, a_t) + \Phi_T(s_T)$ with respect to $a_t$ : $a_t \leftarrow a_t - \eta \nabla_{a_t} J$}.  
\begin{proof}
    Notice that $s_{t+1} = f(a_t, s_t)$ in Eq. \eqref{eq:optimization}. Hence, $\nabla_{s_t} f(s_t, a_t) = \frac{\partial s_{t+1}}{\partial s_{t}}$ . Based on Eq. \eqref{eq:costate},  we can do the backward recursion:
    \begin{align}
        \lambda_t  &= \nabla_{s_t} c(s_t, a_t) + \nabla_{s_t}f(s_t, a_t)\lambda_{t+1} \\
        & = \frac{\partial c(s_t, a_t)}{\partial s_t} + \frac{\partial f(s_t, a_t)}{\partial s_t} \lambda_{t+1}.
    \end{align}
    First consider $t=T-1$:
    \begin{align}
        \lambda_{T-1}  & = \frac{\partial c(s_{T-1}, a_{T-1})}{\partial s_{T-1}} + \frac{\partial f(s_{T-1}, a_{T-1})}{\partial s_{T-1}} \lambda_{T} \\
        & = \frac{\partial c(s_{T-1}, a_{T-1})}{\partial s_{T-1}} + \frac{\partial s_T}{\partial s_{T-1}} \frac{\partial \Phi(s_T)}{\partial s_T} \\ 
        & = \frac{\partial c(s_{T-1}, a_{T-1})}{\partial s_{T-1}} + \frac{\partial \Phi(s_T)}{\partial s_{T-1}};
    \end{align}
    Then consider $t=T-2$:
    \begin{align}
        \lambda_{T-2}  & = \frac{\partial c(s_{T-2}, a_{T-2})}{\partial s_{T-2}} + \frac{\partial f(s_{T-2}, a_{T-2})}{\partial s_{T-2}} \lambda_{T-1} \\
        & = \frac{\partial c(s_{T-2}, a_{T-2})}{\partial s_{T-2}} + \frac{\partial s_{T-1}}{\partial s_{T-2}} (\frac{\partial c(s_{T-1}, a_{T-1})}{\partial s_{T-1}} + \frac{\partial \Phi(s_T)}{\partial s_{T-1}}) \\ 
        & = \frac{\partial c(s_{T-2}, a_{T-2})}{\partial s_{T-2}} + \frac{\partial c(s_{T-1}, a_{T-1})}{\partial s_{T-2}}  +\frac{\partial \Phi(s_T)}{\partial s_{T-2}};
    \end{align}
    and so on for $t =T-3,....,0$. Hence we have
    \begin{align}
        \lambda_t = \sum_{i=t}^{T-1}(\frac{\partial c(s_i, a_i)}{\partial s_t}) + \frac{\partial \Phi(s_T)}{\partial s_t}, \quad  0 \leq t \leq T-1.
    \end{align}
    Hence, the gradient of Hamiltonian $H(s_t, a_t, \lambda_{t+1})$ w.r.t. $a_t$ is:
\begin{align}
    \nabla_{a_t} H(s_t, a_t, \lambda_{t+1}) & = \frac{\partial c(s_t, a_t)}{\partial a_t} + \frac{\partial f(s_t, a_t)}{\partial a_t}\lambda_{t+1} \nonumber\\  
    & = \frac{\partial c(s_t, a_t)}{\partial a_t} + (\sum_{i=t+1}^{T-1} \frac{\partial c(s_t, a_t)}{\partial s_{t+1}} + \frac{\partial \Phi(s_T)}{\partial s_{t+1}}) \frac{\partial s_{t+1}}{\partial a_t} \nonumber \\
    &= \sum_{i=t}^{T-1} \frac{\partial c(s_t, a_t)}{\partial a_{t}} + \frac{\partial \Phi(s_T)}{\partial a_{t}}.
\end{align}
Notice that $c(s_i, a_i)$ for $i=0, ...,t-1$ has no relation to $a_t$, $\sum_{i=1}^{t-1} \frac{\partial c(s_i, a_i))}{\partial a_t} = 0$. Hence, $\nabla_{a_t} J = 0 +  \sum_{i=t}^{T-1} \frac{\partial c(s_t, a_t)}{\partial a_{t}} + \frac{\partial \Phi(s_T)}{\partial a_{t}}$ = $\nabla_{a_t} H$. 
\end{proof}

\subsection{Proof of Theorem 4.2}
\label{proofT1}
\paragraph{Theorem 4.2.} \textit{(Convergence of HAC) Suppose the following hold: (i) $J$ is differentiable  and Lipschitz continuous with constant $L_m$ w.r.t. $a_t$ (ii) $\nabla_{a_t} J$ is Lipschitz continuous with contant $L_J$ (iii) $\pi_{\theta}$ is differentiable  and Lipschitz continuous with constant $L_G$ w.r.t. $\theta$ (iv) $\nabla_{\theta} \pi_{\theta}$ is $L_p$-Lipschitz continuous. The policy agent updating $\theta \leftarrow \theta - \eta \nabla_{\theta}L_{actor}(\theta)$ in HAC guarantees convergence, with learning rate $\eta < 2/(L_JL_G^2+L_pL_m)$}. 
\begin{proof}
    With the actor loss $L_{actor}(\theta) =\frac{1}{K}\sum_{t=0}^{K-1} H(\hat{s}_t, \hat{a}_t, \hat{\lambda}_{t+1})$, the gradient:
    \begin{align}
        \nabla_{\theta}L_{actor}  = \frac{1}{K}\sum_{t=0}^{K-1} \nabla_{\theta}H.
    \end{align}
    
    By the chain rule, we have: 
    \begin{align}
        \nabla_{\theta}H = \nabla_{\hat{a}_t} H \cdot \nabla_{\theta}  \hat{a}_t.
    \end{align}

    From \textit{Lemma 1}, we know:
    \begin{equation}
        \nabla_{\hat{a}_t} H = \nabla_{\hat{a}_t} J.
    \end{equation}

    so it follows that:
    \begin{equation}
        \nabla_{\theta}L_{actor} = \frac{1}{K}\sum_{t=0}^{K-1} \nabla_{\hat{a}_t}J \cdot \nabla_{\theta} \hat{a}_t = \frac{1}{K}\sum_{t=0}^{K-1} \nabla_{\hat{a}_t}J \cdot \nabla_{\theta} \pi_{\theta}(\hat{s}_t)  = \frac{1}{K}\sum_{t=0}^{K-1} \nabla_{\theta} J(\pi_{\theta}(\hat{s}_t)).
    \end{equation}
    

    To prove convergence, we show that $\nabla_{\theta} J(\pi_{\theta}(\hat{s}_t))$ is Lipschitz continuous with respect to $\theta$. For arbitrary $\theta_1, \theta_2$, consider:
    \begin{equation}
        \Vert \nabla_{\theta}J(\pi(\theta_1)) - \nabla_{\theta}J(\pi(\theta_2)) \Vert
    \end{equation}

    Applying the chain rule:
    \begin{equation}
        \nabla_{\theta}J(\pi(\theta)) = \nabla_aJ(\pi_{\theta}) \cdot \nabla_{\theta} \pi(\theta).
    \end{equation}

    Hence we have:
    \begin{align}
         &\Vert \nabla_{\theta}J(\pi(\theta_1)) - \nabla_{\theta}J(\pi(\theta_2)) \Vert  = \Vert \nabla_{a}J(\pi(\theta_1)) \nabla_{\theta} \pi(\theta_1) - \nabla_{a}J(\pi(\theta_2)) \nabla_{\theta} \pi(\theta_2) \Vert \nonumber\\
        &\leq \Vert \nabla_{a}J(\pi(\theta_1)) - \nabla_{a}J(\pi(\theta_2))\Vert \Vert\nabla_{\theta} \pi(\theta_1)\Vert + \Vert\nabla_{\theta} \pi(\theta_1) - \nabla_{\theta} \pi(\theta_2) \Vert \Vert\nabla_{a}J(\pi(\theta_2)) \Vert.
        \label{eq:ineq1}
    \end{align}
    By Assumptions (ii), (iv) we have:
    \begin{align}
        \Vert \nabla_{a}J(\pi(\theta_1)) - \nabla_{a}J(\pi(\theta_2)\Vert &\leq L_J\Vert \pi(\theta_1)- \pi(\theta_2)\Vert,\\
        \Vert\nabla_{\theta} \pi(\theta_1) - \nabla_{\theta} \pi(\theta_2) \Vert &\leq L_p\Vert \theta_1 - \theta_2 \Vert.
     \end{align}   
     By Assumption (i), (iii) we have:
     \begin{align}
         \Vert \pi(\theta_1)- \pi(\theta_2))\Vert &\leq L_G\Vert \theta_1 - \theta_2 \Vert,\\
         \Vert \nabla_{\theta}  \pi(\theta_1) \Vert 
 &= \lim_{h \rightarrow0} \Vert \frac{\pi(\theta_1+h) - \pi(\theta_1)}{h} \Vert\leq L_G, \\
         \Vert\nabla_{a}J(\pi(\theta_2)) \Vert &=\lim_{h \rightarrow0} \Vert \frac{J(\pi(\theta_2)+h) - J(\pi(\theta_2))}{h} \Vert \leq L_m.
     \end{align}
     Substituting these bounds into the inequality Eq. \eqref{eq:ineq1} yields:
     \begin{align}
         \Vert \nabla_{\theta}J(\pi(\theta_1)) - \nabla_{\theta}J(\pi(\theta_2)) \Vert \leq  (L_G^2L_J+L_pL_m) \Vert \theta_1 - \theta_2 \Vert.
     \end{align}
     Thus, $\nabla_{\theta}L_{actor}$ is Lipschitz continuous with constant $L_{act}=L_G^2L_J+L_pL_m$.

     Consider the standard gradient descent update of HAC:
     \begin{equation}
         \theta_{i+1}  = \theta_{i} - \eta \nabla_{\theta}L_{actor}(\theta).
         \label{eq:gdhac}
     \end{equation}

    Since $\nabla_{\theta}L_{actor}$ is Lipschitz continuous, applying the descent lemma:
     \begin{align}
         L_{actor}(\theta_{i+1}) \leq L_{actor}(\theta_{i}) + \nabla_{\theta_i}L_{actor}(\theta_i)  \cdot (\theta_{i+1} - \theta_{i}) + \frac{L_{act}}{2}\Vert\theta_{i+1} - \theta_{i}\Vert^2 .
         \label{eq:gdl}
      \end{align}

    Substituting the update step from Eq.\eqref{eq:gdhac} into Eq.\eqref{eq:gdl}, we obtain:
    \begin{equation}
        L_{actor}(\theta_{i+1}) \leq L_{actor} (\theta_i) - \eta(1-\frac{L_{act}}{2}\eta) \Vert \nabla_{\theta_i} L_{actor}(\theta_i)\Vert^2.
        \label{eq:sum}
    \end{equation}

    For any learning rate satisfying $0 < \eta < \frac{2}{L_{act}} $, the term $\eta(1-\frac{L_{act}}{2}\eta)$. Therefore, the actor loss is non-increasing:
    \begin{equation}
         L_{actor}(\theta_{i+1})  \leq  L_{actor}(\theta_{i}).
    \end{equation}

    This implies that the actor loss decreases monotonically with each iteration.

    Summing the inequality Eq. \eqref{eq:sum} from iteration $i=0$ to $N-1$ gives:
    \begin{equation}
        \sum_{i=1}^{N-1} \Vert \nabla_{\theta_i} L_{actor}(\theta_i)\Vert^2 \leq \frac{1}{\eta(1-L_{act} \cdot \eta/2)} (L_{actor}(\theta_0) - L_{actor}(\theta_N)).
    \end{equation}

   Assuming the existence of a global minimum $\theta^*$ for $L_{actor}(\theta)$, we have:
    \begin{align}
        \lim_{N \rightarrow \infty} \sum_{i=1}^{N-1} \Vert \nabla_{\theta_i} L_{actor}(\theta_i)\Vert^2 &\leq \lim_{N \rightarrow \infty}  \frac{1}{\eta(1-L_{act} \cdot \eta/2)} (L_{actor}(\theta_0) - L_{actor}(\theta_N)) \nonumber \\
        &\leq \frac{1}{\eta(1-L_{act} \cdot \eta/2)} (L_{actor}(\theta_0) - L_{actor}(\theta^*)) \nonumber \\
        & < \infty.
    \end{align}

    The sum converges. As a result:
    \begin{equation}
        \lim_{i \rightarrow \infty } \Vert \nabla_{\theta_i} L_{actor}(\theta_i)\Vert^2 = 0.
    \end{equation}

    which implies that $\theta_i$ converges to a stationary point corresponding to a local or global minimum where the gradient  $\nabla_{\theta} L_{actor}(\theta)$ 
    vanishes.
    
\end{proof}

\subsection{Proof of Theorem 4.4}
\label{proofT2}

\paragraph{Lemma 4.3.} \textit{(Hamiltonian Conservation Law in Optimal Control) The Hamiltonian of the continuous-time optimal control problem:  $\int_{t=0}^T r(s(t), a(t))+r_T(s(T)), \text{s.t.}\;  \dot{s} = f(s(t), a(t))$, is constant across time when applying the assumed optimal $a(t)$ from PMP}. 

\paragraph{Theorem 4.4.} \textit{(Tighter Bound of Critic Estimation Error in HAC compared with MVE) When HAC and MVE converge to their optimal solutions, suppose: (i) the learned dynamics $\hat{f}$ in both HAC and MVE satisfy the generalization error $\delta$ in $K$-step depth (ii) the error $| Q_{\phi'} - Q|$ between target critic network $Q_{\phi'}(s_t, a_t)$ and ground-truth $Q(s_t, a_t)$ is bounded by $\epsilon$. If: (a) $r(s_t, a_t): = -(s_t^{T}Rs_t+a_t^TPa_t)$ is quadratic; (b) $r_T(s_K):= C$ (Constant),  $As_T+b$ (Linear), or $s_T^{T}As_T$ (Quadratic), we have: $\max 
|H-H_g | \leq \max | Q_{\phi} - Q| $.} 

\begin{proof}
We begin by analyzing the critic estimation error of MVE. The target used to train the critic $Q_\phi$ is:
\begin{equation}
   \sum_{t=1}^{K-1}\gamma^{t-1} r(s_t, a_t) + \gamma^{K} Q_{\phi_{'}}(s_{K}, a_{K}).
\end{equation}

where $Q{\phi'}$ denotes the target critic network. Therefore, the error between the learned critic and ground truth is:
\begin{align}
    |Q_{\phi} - Q| &= |\sum_{t=1}^{K-1}\gamma^{t-1} r(\hat{s}_t, a_t) + \gamma^{K} Q_{\phi_{'}}(\hat{s}_{K}, a_{K}) - \sum_{t=1}^{K-1}\gamma^{t-1} r(s_t, a_t) + \gamma^{K} Q_{\phi_{'}}(s_{K}, a_{K})| \\
    &\leq \sum_{t=0}^{K-1}\gamma^{t} |r(\hat{s}_t, a_t) - r(s_t, a_t)| + \gamma^{K} |Q_{\phi_{'}}(\hat{s}_{K}, a_{K}) - Q_{\phi_{'}}(s_{K}, a_{K})|.
    \label{eq:mveerror}
\end{align}
where$\hat{s}_t$ denotes the predicted state under the learned dynamics $\hat{f}$. 

Substituting $r(s_t, a_t): = -(s_t^{T}Rs_t+a_t^TPa_t)$ into Eq. \eqref{eq:mveerror}:
\begin{align}
   |Q_{\phi} - Q| &\leq \sum_{t=0}^{K-1}\gamma^{t} |\hat{s}_t^TR\hat{s}_t - s_t^TRs_t| + \gamma^{K}|Q_{\phi_{'}}(\hat{s}_{K}, a_{K}) - Q(s_{K}, a_{K})| \nonumber \\
   &\leq  \sum_{t=0}^{K-1}\gamma^{t} |\hat{s}_t^TR\hat{s}_t - s_t^TRs_t| + \gamma^{K}|Q_{\phi_{'}}(\hat{s}_{K}, a_{K}) - Q(\hat{s}_{K}, a_{K})| \nonumber \\
   &+ \gamma^{K}|Q(\hat{s}_{K}, a_{K}) - Q(s_K, a_K)| \nonumber \\
   &\leq \sum_{t=0}^{K-1}\gamma^{t} |\hat{s}_t^TR\hat{s}_t - s_t^TRs_t| + \gamma^{K}\epsilon+\gamma^{K}|Q(\hat{s}_{K}, a_{K}) - Q(s_K, a_K)|.
   \label{eq:mvebound}
\end{align}

Under known terminal reward $r_T$, we can upper-bound the final term $|Q(\hat{s}_{K}, a_{K}) - Q(s_K, a_K)|$ by:
\begin{equation}
    |Q(\hat{s}_{K}, a_{K}) - Q(s_K, a_K)| \leq |r_T(\hat{s}_{K}) - r_T(s_K)|
\end{equation}
Thus:
\begin{equation}
    |Q_{\phi} - Q| \leq \sum_{t=0}^{K-1}\gamma^{t} |\hat{s}_t^TR\hat{s}_t - s_t^TRs_t| + \gamma^{K}\epsilon + \gamma^{K}|r_T(\hat{s}_{K}) - r_T(s_K)|. 
     \label{eq:mveupper}
\end{equation}


Next, consider the critic error in HAC. By the Hamiltonian conservation property (\textit{Lemma 2}), the Hamiltonian is constant over time at convergence ($H(s_t, a_t, \lambda_{t+1}) = H(s_{t-1}, a_{t-1}, \lambda_{t})$). Therefore, we just need to compute the deviation in HAC’s Hamiltonian at time $K-1$:
\begin{align}
    |H_g-H| &= |H(\hat{s}_{K-1}, a_{K-1}, \lambda_{K}) - H(s_{K-1}, a_{K-1}, \lambda_{K})| \nonumber \\
    & = |\gamma^{K-1}(r(\hat{s}_{K-1}, a_{K-1} - r(s_{K-1}, a_{K-1})) + \hat{\lambda}_K\hat{x}_{K} - \lambda_Kx_K| \nonumber \\
    &= |\gamma^{K-1}(\hat{s}_{K-1}^TR\hat{s}_{K-1} - s_{K-1}^TRs_{K-1}) + \hat{\lambda}_K\hat{x}_{K} - \lambda_Kx_K|.
    \label{eq:HACbound}
\end{align}

We analyze the terminal adjoint state $\lambda_K$ under each terminal reward $r_T$ form:

\textbf{Case 1: Constant $r_T:=C$}.

Based on Pontryagin Maximum Principles (PMP), if $r_T=C$, then $\lambda_K = \frac{\partial \Phi(s_K)}{\partial s_K} = - \frac{\partial r_T(s_K)}{\partial s_K} = 0$. Substituting into into Eq. \eqref{eq:HACbound}:
\begin{align}
    |H_g-H| = |\gamma^{K-1}(\hat{s}_{K-1}^TR\hat{s}_{K-1} - s_{K-1}^TRs_{K-1})|.
    \label{eq:con}
\end{align}

\textbf{Case 2: Linear $r_T:=As_K+b$}.

Based on PMP, $r_T(s_K) = A s_K + b$, then $\lambda_K = - \frac{\partial r_T(s_K)}{\partial s_K} = -A$. Substituting into into Eq. \eqref{eq:HACbound}:
\begin{align}
    |H_g-H| &\leq |\gamma^{K-1}(\hat{s}_{K-1}^TR\hat{s}_{K-1} - s_{K-1}^TRs_{K-1})| + |A(\hat{s}_K-s_K)| \nonumber \\
    &\leq |\gamma^{K-1}(\hat{s}_{K-1}^TR\hat{s}_{K-1} - s_{K-1}^TRs_{K-1})| + a\delta .
    \label{eq:lin}
\end{align}
where $a$ is an eigenvalue of constant matrix $A$.

\textbf{Case 3: Quadratic $r_T:=s^T_KAs_K$}.

Based on PMP, $r_T=s^T_KAs_K$, then $\lambda_K = - \frac{\partial r_T(s_K)}{\partial s_K} = -2As_K$
Substituting in Eq. \eqref{eq:HACbound}:
\begin{align}
    |H_g-H| &= |\gamma^{K-1}(\hat{s}_{K-1}^TR\hat{s}_{K-1} - s_{K-1}^TRs_{K-1}) + \hat{\lambda}_K\hat{s}_{K} - \lambda_Ks_K)| \nonumber \\ 
    &\leq  |\gamma^{K-1}(\hat{s}_{K-1}^TR\hat{s}_{K-1} - s_{K-1}^TRs_{K-1})|  + |\hat{s}_{K}^TA \hat{s}_{K}- s_K^TAs_K|.
    \label{eq:qua}
\end{align}

Now compare the maximum critic estimation error in HAC $|H_g-H|$ and in MVE $|Q_{\phi} - Q|$ in all cases. For constant case, based on Eq. \eqref{eq:mveupper} and Eq. \eqref{eq:con}:
\begin{align}
    &\max |Q_{\phi} - Q| - \max|H_g-H| = \sum_{t=0}^{K-2}\gamma^{t} |\hat{s}_t^TR\hat{s}_t - s_t^TRs_t| + \gamma^{K}\epsilon > 0.
\end{align}

For linear and quadratic cases, suppose $r_T$ is also discounted by $\gamma$. For linear case, based on Eq. \eqref{eq:mveupper} and Eq. \eqref{eq:lin}:
\begin{align}
    &\max |Q_{\phi} - Q| - \max|H_g-H|  \nonumber \\
    &= \sum_{t=0}^{K-2}\gamma^{t} |\hat{s}_t^TR\hat{s}_t - s_t^TRs_t| + \gamma^{K}\epsilon+\gamma^{K}|r_T(\hat{s}_{K}) - r_T(s_K)| - a\delta \nonumber \\
    &= \sum_{t=0}^{K-2}\gamma^{t} |\hat{s}_t^TR\hat{s}_t - s_t^TRs_t| + \gamma^{K}\epsilon+ \gamma^{K}|A(\hat{s}_{K} - s_K)| - \gamma^{K}a\delta = 0 \nonumber \\
    &= \sum_{t=0}^{K-2}\gamma^{t} |\hat{s}_t^TR\hat{s}_t - s_t^TRs_t| +  \gamma^{K}\epsilon + \gamma^{K}a\delta - \gamma^{K}a\delta \nonumber \\
    &= \sum_{t=0}^{K-2}\gamma^{t} |\hat{s}_t^TR\hat{s}_t - s_t^TRs_t| + \gamma^{K}\epsilon>0.
\end{align}

For quadratic case, based on Eq. \eqref{eq:mveupper} and Eq. \eqref{eq:qua}:
\begin{align}
    &\max |Q_{\phi} - Q| - \max|H_g-H|  \nonumber \\
    &=\sum_{t=0}^{K-2}\gamma^{t} |\hat{s}_t^TR\hat{s}_t - s_t^TRs_t| + \gamma^{K}\epsilon+\gamma^{K}|r_T(\hat{s}_{K}) - r_T(s_K)| - \gamma^{K} |\hat{s}_{K}^TA \hat{s}_{K}- s_K^TAs_K| \nonumber  \\
    &= \sum_{t=0}^{K-2}\gamma^{t} |\hat{s}_t^TR\hat{s}_t - s_t^TRs_t| + \gamma^{K}\epsilon + \gamma^{K}|\hat{s}_K^TA\hat{s}_K - s_K^TAs_K| - \gamma^{K} |\hat{s}_{K}^TA \hat{s}_{K}- s_K^TAs_K| \nonumber \\
    &= \sum_{t=0}^{K-2}\gamma^{t} |\hat{s}_t^TR\hat{s}_t - s_t^TRs_t| + \gamma^{K}\epsilon>0.
\end{align}
Thus, we conclude:
\begin{equation}
    \max |Q_{\phi} - Q| - \max|H_g-H| >0.
\end{equation}
This holds under constant, linear, and quadratic terminal reward $r_T$.
\end{proof}

\subsection{Error Analysis of HAC when $K<T$}
\label{appendix::error_analyze}

For optimal control system: 
\begin{equation*}
    \begin{aligned}
        \min  \quad &J=\sum_{t=0}^{T-1} c(s_t, a_t) + \Phi(s_T)\\
       s.t. \quad &s_{t+1}=f(s_t, a_t).
    \end{aligned}
\end{equation*}
Considering the PMP conditions: 
\begin{equation}
    \lambda_t = \nabla_{s_t} H(s_t, a_t, \lambda_{t+1}) = \frac{\partial c(s_t, a_t)}{\partial s_t} + \lambda_{t+1}^\top \frac{\partial f(s_t, a_t)}{\partial s_t}.
    \label{eq::appendix_costate}
\end{equation}

In the ground-truth setting (where there is no dynamics error), PMP is applied over the entire $T$-step control horizon, and the resulting terminal state $\lambda^*_{T}$ is given by:
\begin{equation}
    \lambda^*_{T} = \frac{\partial \Phi(s_T)}{\partial s_T}.
\end{equation}

Hence, by applying Eq.~\eqref{eq::appendix_costate}, the ground-truth costate $\lambda^*_{t+K}$ can be computed by:
\begin{equation}
\begin{aligned}
\lambda^{*}_{t+K}=
\sum_{i=t+K}^{T-1}
\left(
\prod_{j=t+K}^{i-1}
\left( \frac{\partial f(s_j,a_j)}{\partial s_j} \right)^{\!\top}
\right)
\frac{\partial c(s_i,a_i)}{\partial s_i}  +
\left(
\prod_{j=t+K}^{T-1}
\left( \frac{\partial f(s_j,a_j)}{\partial s_j} \right)^{\!\top}
\right)
\lambda^*_{T}.
\end{aligned}
\label{eq::appendix_true}
\end{equation}

However, in HAC, PMP is applied over a $K$-step truncated imaginary rollout on $[t, t+K]$, yielding the estimated costate $\hat{\lambda}_{t+K}$ as:
\begin{equation}
    \hat{\lambda}_{t+K} = \frac{\partial \Phi(s_{t+K})}{\partial s_{t+K}}.
    \label{eq::appendix_estimate}
\end{equation}

This expression differs from the ground-truth costate $\lambda^{*}_{t+K}$ given in Eq.~\eqref{eq::appendix_true}. 
Since the Hamiltonian is defined as
$H(s_t,a_t,\lambda_{t+1}) = c(s_t,a_t) + \lambda_{t+1}^\top f(s_t,a_t)$
such errors in the estimated costate when $K<T$ directly induces an approximation error in the Hamiltonian used by HAC, hence can influence the performance of HAC.

However, in certain special cases, the influence of costate estimation error induced by the short-horizon $K$-step imaginary rollout can be substantially mitigated. 
Specifically, suppose that (i) the running cost is state-independent, i.e.,
$\frac{\partial c(s_i,a_i)}{\partial s_i} = 0$, and 
(ii) the terminal cost $\Phi(s_T)$ is linear, such that
$\frac{\partial \Phi(s_T)}{\partial s_T} = C$, where $C$ is a constant vector.
Under these assumptions, Eqs.~\eqref{eq::appendix_true} and \eqref{eq::appendix_estimate} yield:
\begin{equation*}
    \begin{aligned}
        \hat{\lambda}_{t+K} &= C \\
        \lambda^{*}_{t+K}&=\left(
\prod_{j=t+K}^{T-1}
\left( \frac{\partial f(s_j,a_j)}{\partial s_j} \right)^{\!\top}
\right)
C
    \end{aligned}
\end{equation*}

The ground-truth costate $\lambda^{*}{t+K}$ is proportional to the estimated costate $\hat{\lambda}{t+K}$. Specifically, we denote $\lambda^{*}_{t+K} = \beta \hat{\lambda}_{t+K} = \beta C$.Substituting this relation into Eq.~\eqref{eq::appendix_costate}, we obtain:
\begin{equation*}
\begin{aligned}
\hat{\lambda}_\tau
=
\left(
\prod_{i=\tau}^{t+K-1}
\left( \frac{\partial f(s_i,a_i)}{\partial s_i} \right)^{\!\top}
\right)
\hat{\lambda}_{t+K},
\qquad \tau \in [t, t+K]. \\
\lambda^{*}_{\tau}
=
\left(
\prod_{i=\tau}^{t+K-1}
\left( \frac{\partial f(s_i,a_i)}{\partial s_i} \right)^{\!\top}
\right)
\lambda^{*}_{t+K},
\qquad \tau \in [t, t+K]. 
\end{aligned}
\end{equation*}
According to PMP, the optimal condition is computed by: 
\begin{equation}
    \frac{\partial H(s_t, a_t, \lambda_{t+1})}{\partial a_t} = \frac{\partial c(s_t, a_t)}{\partial a_t} +  \lambda_{t+1}^T \frac{\partial f(s_t, a_t)}{\partial a_t}.
\end{equation}

Hence, the ground-truth optimal action $a_t^*$ and the estimated optimal action $\hat{a}_t$ obtained from the $K$-step rollout satisfy:
\begin{equation*}
    \begin{aligned}
        \frac{\partial c(s_t, a^*_t)}{\partial a^*_t} +  \lambda_{t+1}^{*T} \frac{\partial f(s_t, a^*_t)}{\partial a^*_t} = 0, \\
\frac{\partial c(s_t, \hat{a}_t)}{\partial \hat{a}_t} + \hat{\lambda}_{t+1}^{T} \frac{\partial f(s_t, \hat{a}_t)}{\partial \hat{a}_t} = 0.\\
    \end{aligned}
\end{equation*}
Therefore:
\begin{equation*}
    \begin{aligned}
        \frac{\partial c(s_\tau, a^*_\tau)}{\partial a^*_\tau} +  \beta C^T \frac{\partial f(s_\tau, a^*_\tau)}{\partial a^*_\tau} = 0, \qquad \tau \in [t, t+K-1].\\
\frac{\partial c(s_\tau, \hat{a}_\tau)}{\partial \hat{a}_\tau} + C^T \frac{\partial f(s_\tau, \hat{a}_\tau)}{\partial \hat{a}_\tau} = 0,\qquad \tau \in [t, t+K-1].\\
    \end{aligned}
\end{equation*}

Hence, when $\beta \approx 1$ (i.e., the system is \emph{marginally stable}), the error in the obtained optimal solution can be neglected.

\subsection{Potential Extension to Stochastic Dynamics Systems}
\label{app::stochastic}
The current formulation of HAC is designed for deterministic environments and is not directly applicable to stochastic dynamics. 
Nevertheless, the framework can be naturally extended to stochastic settings by reformulating the objective as the minimization of the expected cost,
\begin{equation}
\mathbb{E}_{\mathcal{D}}\!\left[
J = \sum_{t=0}^{T-1} c(s_t,a_t) + \Phi_T(s_T)
\right],
\end{equation}
where the expectation is approximated via Monte Carlo sampling over $N$ trajectories.

As shown in \textit{Lemma}~\ref{appendix::MPMP}, under this formulation the Mean Pontryagin Maximum Principle (MPMP) holds, in which the expected Hamiltonian
\begin{equation}
\mathbb{E}_D\!\left[ H(s_t,a_t,\lambda_{t+1}) \right]
\end{equation}
plays an analogous role to the Hamiltonian in the deterministic HAC framework. 
Consequently, our main theoretical results (\textit{Theorem}~\ref{theorem::Convergence} and~\ref{theorem::bound})will  extend to stochastic environments by replacing deterministic quantities with their Monte Carlo estimates, thereby preserving the core theoretical foundations of HAC.
\begin{lemma}
    (Discrete-time Mean Pontryagin Maximum Principles). Let $a^*_t$ be an optimal control that minimize the mean cost $\mathbb{E}_{\mathcal{D}} [J = \sum_{t=0}^{T-1} c(s_t, a_t) + \Phi_T(s_T)]$ (where $\mathcal{D}$ contains $i=1,...,N$ sample trajectories). Then a necessary condition for $a^*_t$ is: $\nabla_{a^*_t} \mathbb{E}_{D}[H(s^*_t, a^*_t, \lambda^*_{t+1})] = 0$.
    \label{appendix::MPMP}
\end{lemma}
\begin{proof}
    We have $\mathbb{E}_{\mathcal{D}} [J = \sum_{t=0}^{T-1} c(s_t, a_t) + \Phi_T(s_T)] = \sum_{i=1}^{N} (\sum_{t=0}^{T-1} c(s^i_t, a^i_t) + \Phi_T(s^i_T))$.
    
   Consider a perturbation of the control $ a_t^{i,\epsilon} = a_t^i + \epsilon \delta a_t^i $ and define the perturbed cost:
   \begin{align}
       J_N^\epsilon = \frac{1}{N} \sum_{i=1}^N \left[ \sum_{t=0}^{T-1} c(s_t^{i,\epsilon}, a_t^{i,\epsilon}) + \Phi(s_T^{i,\epsilon}) \right]
   \end{align}

Differentiating \( J_N^\epsilon \) with respect to \( \epsilon \) at \( \epsilon = 0 \) gives:
\begin{align}
    \left. \frac{dJ_N}{d\epsilon} \right|_{\epsilon = 0} =
\frac{1}{N} \sum_{i=1}^N \left[ 
\sum_{t=0}^{T-1} \left(
\nabla_{s_t^i} c(s_t^i, a_t^i)^\top \frac{ds_t^i}{d\epsilon}
+ \nabla_{a_t^i} c(s_t^i, a_t^i)^\top \delta a_t^i
\right)
+ \nabla_{s_T^i} \Phi(s_T^i)^\top \frac{ds_T^i}{d\epsilon}
\right]
\end{align}

The state perturbation propagates forward via:
\begin{align}
    \frac{ds_{t+1}^i}{d\epsilon} = \frac{\partial f}{\partial s_t^i} \cdot \frac{ds_t^i}{d\epsilon}
+ \frac{\partial f}{\partial a_t^i} \cdot \delta a_t^i
\end{align}

Since Trajectory $i$ is deterministic, to eliminate \( \frac{ds_t^i}{d\epsilon} \), we can perform costates backward in time accorrding to Pontryagin Maximum Principles:
\begin{align}
    \lambda_T^i := \nabla_{s_T^i} \Phi(s_T^i), \quad 
\lambda_t^i := \left( \frac{\partial f}{\partial s_t^i} \right)^\top \lambda_{t+1}^i + \nabla_{s_t^i} c(s_t^i, a_t^i)
\end{align}

Substituting these into the variation yields:
\begin{align}
    \left. \frac{dJ_N}{d\epsilon} \right|_{\epsilon = 0} =
\frac{1}{N} \sum_{i=1}^N \sum_{t=0}^{T-1} 
\nabla_{a_t^i} \mathcal{H}^i(s_t^i, a_t^i, \lambda_{t+1}^i)^\top \delta a_t^i
\end{align}

Optimality requires $\left. \frac{dJ_N}{d\epsilon} \right|_{\epsilon = 0} =0$. Since the perturbation \( \delta a_t^i \) is arbitrary, we have:
\begin{equation}
    \nabla_{a_t^i} \mathcal{H}^i(s_t^i, a_t^i, \lambda_{t+1}^i) = 0
\end{equation}

Hence:
\begin{align}
    \nabla_{a^*_t} \mathbb{E}_{D}[H(s^*_t, a^*_t, \lambda^*_{t+1})] = 0
\end{align}

\end{proof}

\section{Experiment Settings}
\label{expset}
In this section, we detail the experimental settings used across all five environments. All experiments were conducted on a Google Colaboratory instance equipped with a T4 GPU and 12 GB of RAM. For consistency across environments, we employ a 3-layer Multi-Layer Perceptron (MLP) architecture (each hidden layer has 128 neurons) for the dynamics model, critic network, and policy agent learning.

\paragraph{LQR:} The Linear Quadratic Regulator (LQR) is a fundamental problem in optimal control theory. In our experiments, we implement a short-horizon 10-step LQR, with state dimension $m=5$ and action dimension $n=3$.  The ground-truth system dynamics follow a linear model: $s_{t+1} = As_t+Ba_t$, where we set $A=I$ (the identity matrix) and $ B= \begin{bmatrix}
        1 & 0 & 0 \\
        0 & 1 & 0\\
        0 & 0 & 1 \\
        1 & 1 & 0 \\
        0 & 1 & 1 \\
    \end{bmatrix}$). The running reward $r(s_t, a_t)$ is defined as a quadratic function: $-r(s_t, a_t) = s_t^TRs_t+a_t^TPa_t$, and the terminal reward $r_T$ is also quadratic: $-r_T(s_K) = s_K^TAs_K$. In our setup, we use identity matrices for both $R, P$ and set the terminal cost matrix $A=0.1I$. The initial state is given by $s_0 = [0, 1, 1, 0, 0]$ with added Gaussian noises. For offline RL settings, we created the pre-collected datasets with 5,000 random transition samples.

\paragraph{Pendulum:} We employ a short-horizon 10-step Pendulum environment with state dimension $m=2$ and action dimension $n=1$. The action $a_{t}$ represents the applied torque, while the 2-dimensional state is defined as $s_t=[q_t, dq_t]$ where $q_t$ denotes angular velocity and $dq_t$ denotes angular acceleration. The running reward $r(s_t, a_t)$ is defined as quadratic: $-r(s_t, a_t) = w_q(q_t-q_f)^2+w_{dq}(dq_t-dq_f)^2+w_ua_t^2$, where $q_f$ and $dq_f$ are target angular velocity and acceleration for $s_t$ and $w_q=10$, $w_{dq}=1$, $w_u=0.1$ are corresponding cost weights. The terminal reward $r_T$ is also quadratic: $-r_T(s_K) = w_q(q_K-q_f)^2+w_{dq}(dq_K-dq_f)^2$. The system dynamics $f(s_{t}, a_{t})$ are derived from \textit{Newton’s Second Law of Rotation} and defined as:
\begin{equation}
    f(s_{t}, a_{t}) = s_t + \Delta [dq_t, \frac{a_t-mglq_t-\sigma sin(q_t)}{I}];
\end{equation}
where $m$ and $l$ denotes denote the pendulum’s mass and length; $g$ is the gravitational constant; $\sigma$ is the damping ratio; $\Delta$ is the time step length, and $I=\frac{1}{3}mgl^2$ is the pendulum's moment of inertia. The initial state is given by $s_0=[0, 0]$ with added Gaussian noises. Compared to the LQR environment, Pendulum exhibits greater dynamic complexity due to the presence of nonlinear trigonometric terms in dynamics, which increase the difficulty of achieving optimal solution.  For offline RL settings, we created the pre-collected datasets with 20,000 random transition samples.

\paragraph{MountainCar:} We evaluate our method on the long-horizon \textit{MountainCar} environment from \textit{Gymnasium}, setting the episode length to be 200 timesteps. \textit{MountainCar} is a deterministic Markov decision process in which a car is initialized stochastically at the bottom of a sinusoidal valley \cite{Moore90efficientmemory-based}. The agent controls the car using a continuous one-dimensional acceleration action ($n=1$). The state space has dimension $m=2$, where the first state component $s_t[0]$ denotes the car’s horizontal (x-axis) position. The objective is to accelerate strategically so as to reach the goal position at $x=0.45$ on the right hill. To align with the assumptions of \textit{Theorem}~\ref{theorem::bound}, we adopt the continuous-action variant and modify the reward structure by using a quadratic running cost $r(s_t,a_t)=-w_a a_t^\top a_t$ with $w_a=0.1$, and a linear terminal reward $r_T(s_K)=w_s\bigl(s_K[0]-0.45\bigr)$ with $w_s=100$.  For offline RL settings, we created the pre-collected datasets with 200,000 random transition samples.

\paragraph{Swimmer:} We evaluate our method on the long-horizon \textit{MuJoCo} \textit{Swimmer} environment, configured with a 1000-step episode length. The environment has a state dimension $m=10$ and an action dimension $n=2$. The first element of state $s_t[0]$ represents the robot's x-coordinate (position). The ground-truth system dynamics are not explicitly provided in this environment. The running reward $r(s_t, a_t)$ is defined as a quadratic penalty on the control effort $r(s_t, a_t) =- w_ca^T_ta_t$ and the terminal reward is defined as a linear function: $r_T(s_K) = w_f(s_K[0]-s_0[0])$, where $w_c=10^{-4}$ and $w_f = 100$ are predefined reward weights. This environment serves to evaluate the performance of our HAC method in high-dimensional, long-horizon control tasks.  For offline RL settings, we created the pre-collected datasets with 500,000 random transition samples.

\paragraph{Hopper:} Consistent with the setup in the \textit{Swimmer} environment, we also evaluate our method on the \textit{MuJoCo} \textit{Hopper} environment with a long-horizon of 1000 timesteps, state dimension $m=12$, and action dimension $n=3$. The first element of state $s_t[0]$ denotes the robot's x-coordinate (position). As with \textit{Swimmer}, the true system dynamics are not explicitly disclosed. The terminal reward retains a linear structure: $r_T(s_K) = w_f(s_K[0]-s_0[0])$, where $w_f = 500$ is a predefined weight that incentivizes forward motion. However, in contrast to \textit{Swimmer}, the running reward $r(s_t, a_t)$ incorporates two components: a quadratic penalty $- w_ca^T_ta_t$ on the control input and a binary "healthy reward" $He$: $r(s_t, a_t) = - w_ca^T_ta_t + He$, where $w_c=10^{-3}$. The term $He$ is a health indicator: the agent receives $He=1$ if the system remains within a designated safe operational domain (e.g. all components in $s_t$ lie within the interval $[-100, 100]$); Otherwise, the episode terminates immediately. This environment is used to assess the effectiveness of our HAC method in settings that deviate from the assumptions outlined in \textit{Theorem}~\ref{theorem::bound}.  For offline RL settings, we created the pre-collected datasets with 500,000 random transition samples.

\paragraph{LQR OOD Simulation Settings:} We conduct a simple experiment to evaluate the robustness of our HAC to out-of-distribution (OOD) issues on the 10-step LQR environment, with state dimension $m=5$ and action dimension $n=3$. During critic training, the initial state is set to $[0, 1, 1, 0, 0]$ with added noise, while during testing, it is shifted to $[0, 0, 1, 1, 0]$, also with noise. We compare the optimal trajectories generated by MVE-DDPG and our HAC method under both the original and shifted initial state conditions. 

\section{Additional Experiment Results}
\label{abl}

\subsection{Critic Estimation Error} 
\label{app::critic_error}
We evaluate the critic estimation error of HAC to assess whether it improves critic accuracy compared to MVE. Following the evaluation protocol of MVE-DDPG, we replicate the critic estimation plots for both methods on the LQR task. In this setting, the ground-truth Q-function $Q_t(s_t, a_t)$ and Hamiltonian $H(s_t, a_t, \lambda_{t+1})$ can be computed analytically, enabling precise error quantification.

For the standard LQR problem:
\begin{equation*}
\begin{aligned}
\min_{\{a_t\}_{t=0}^{T-1}} \quad
& J
= \sum_{t=0}^{T-1}
\left(
s_t^\top U s_t + a_t^\top R a_t
\right)
+ s_T^\top U_f s_T, \\
\text{s.t.} \quad
& s_{t+1} = A s_t + B a_t, \qquad x_0 \ \text{given}, \\
& U \succeq 0,\; R \succ 0,\; U_f \succeq 0 .
\end{aligned}
\end{equation*}

According to the \textit{Algebraic Riccati Equation}, we have the ground-truth Q function $Q_t(s_t, a_t)$:
\begin{equation*}
    \begin{aligned}
V_t(s) &= s^\top P_t s, \\[6pt]
Q_t(s,a)
&= s_t^\top U s_t + a_t^\top R a_t
+ (A s_t + B a_t)^\top P_{t+1} (A s_t + B a_t).
\end{aligned}
\end{equation*}

while the ground-truth Hamiltonian equals to:
\begin{equation*}
    \begin{aligned}
    \lambda_t
&= \frac{\partial}{\partial s_t}
\left(
s_t^\top P_t s_t
\right)
= 2 P_t s_t,\\
H(s_t, a_t, \lambda_{t+1})
&= \min_{u_t}
\left[
s_t^\top U st
+ a_t^\top R a_t
+ (A s_t + B a_t)^\top P_{t+1} (A s_t + B a_t)
\right] \\[6pt]
&= x_t^\top
\left(
U
+ A^\top P_{t+1} A
- A^\top P_{t+1} B
\left( R + B^\top P_{t+1} B \right)^{-1}
B^\top P_{t+1} A
\right)
s_t .
    \end{aligned}
\end{equation*}
where $P_t$ satisfies:
\begin{equation*}
    \begin{aligned}
P_t
= U
+ A^\top P_{t+1} A
- A^\top P_{t+1} B
\left( R + B^\top P_{t+1} B \right)^{-1}
B^\top P_{t+1} A ,
\qquad P_T = U_f .
    \end{aligned}
\end{equation*}

Based on these analytical expressions, we compute the estimation error with respect to the ground-truth Q-function and Hamiltonian in the LQR environment. We collect 10 randomly generated trajectories and, for each transition $(s_t, a_t)$, compute the normalized estimated Q-value $\hat{Q}(s_t, a_t)$ and the corresponding normalized ground-truth value $Q^*(s_t, a_t)$ for MVE-DDPG. Similarly, for HAC, we compute the normalized estimated Hamiltonian $\hat{H}(s_t, a_t, \lambda_{t+1})$ and the ground-truth Hamiltonian $H^*(s_t, a_t, \lambda_{t+1})$. Each transition is plotted as a point with $(Q^*, \hat{Q})$ or $(H^*, \hat{H})$ as the x- and y-coordinates, respectively. The resulting comparisons are shown in Fig.~\ref{fig:CriticErrors}.

\begin{figure}[h]
    \centering
    \begin{subfigure}[b]{0.5\linewidth}
        \centering
        \includegraphics[width=\linewidth]{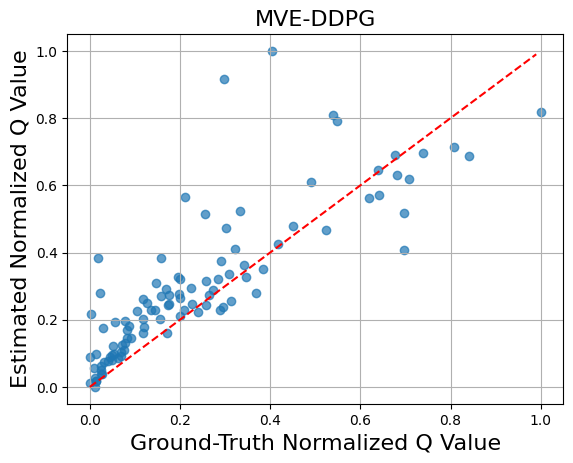}
        \caption{Q-error in MVE-DDPG}
        \label{fig:QError}
    \end{subfigure}%
    \hfill
    \begin{subfigure}[b]{0.5\linewidth}
        \centering
        \includegraphics[width=\linewidth]{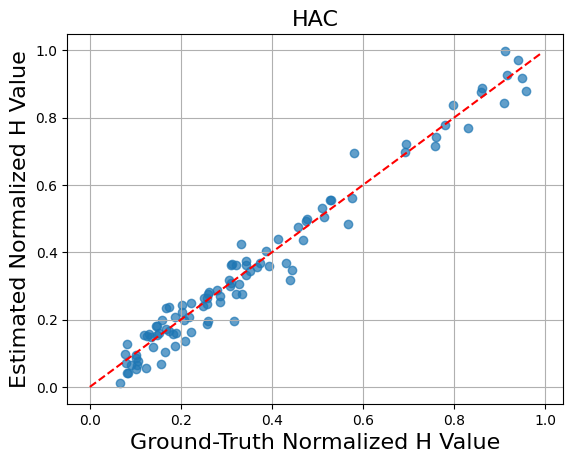}
        \caption{H-error in HAC.}
        \label{fig:HError}
    \end{subfigure}
    \caption{\textbf{Critic Estimation Comparison:} Plots comparing normalized estimated Q-values (in MVE-DDPG) and Hamiltonian (in our HAC) against ground-truth sampled from 10 random test episodes in LQR. An ideal critic should concentrates on the red unity line.}
    \label{fig:CriticErrors}
\end{figure} 

Since an ideal critic would yield estimates lying exactly on the unity line (i.e., $Q^* = \hat{Q}$ and $H^* = \hat{H}$), the results empirically support the bound in Theorem~\ref{theorem::bound}. In particular, the distribution of test points $(H^*, \hat{H})$ produced by HAC is more tightly concentrated around the unity line, indicating that HAC achieves more accurate critic estimation.

\subsection{Ablation Study on Ensemble Number} \textit{Theorem}~\ref{theorem::bound} demonstrates that the critic estimation error in MVE is consistently larger than that in our proposed HAC method, which can negatively impact the performance of MVE-based algorithms. A widely adopted strategy to address this limitation in MVE is to use an ensemble of critic networks. By introducing $N$ critic functions and minimizing over their outputs, the critic estimation becomes more accurate and helps mitigate the overestimation bias, thereby supporting state-of-the-art performance \cite{chen2021randomized}.   This trick has been successfully employed in SAC \cite{haarnoja2018soft}, where a lightweight ensemble of $N=2$ Q-networks is used, with value estimation taken as the minimum across them. Similarly, enhancing MVE-DDPG by simply incorporating such critic ensembles can improve its performance to a level comparable with our HAC method.

While ensemble critic networks can boost MVE-DDPG's performance, they significantly increase computational cost. Each critic in the ensemble requires updating independently, which substantially increases training time~\cite{nikulin2022q}. Similarly, HAC also incurs additional computational overhead due to its backward recursive costate computation, in contrast to MVE’s simpler forward-pass for multi-step return estimation. In this ablation study, we aim to assess the computational efficiency of our HAC method in comparison to MVE-DDPG enhanced with critic ensembles. Specifically, we seek to answer two key questions: (1) How many ensemble number is required for MVE-DDPG to match the performance of HAC? (2) Given similar performance levels, how much computational time can HAC save? We evaluate these questions on the long-horizon task (\textit{Swimmer}). Following prior works~\cite{huang2017learning, parkmodel, nikulin2022q}, we vary the number of critic networks in the ensemble from $N=2$ to $N=5$. Computational efficiency is measured in terms of the average time per gradient step during training, consistent with the evaluation index mentioned in~\cite{nikulin2022q}.

\begin{table}[h!]
\centering
\begin{tabular}{c|c|c|c|c}
\hline
\textbf{Ensemble number} & \multicolumn{2}{|c|}{\textbf{Performance}} & \multicolumn{2}{|c}{\textbf{Computation time per gradient step (s)}} \\
\cline{2-5}
                         & \textbf{MVE-DDPG} & \textbf{HAC} & \textbf{MVE-DDPG} & \textbf{HAC} \\
\hline
None                     & 37.491           & \multirow{3}{*}{44.366} & 0.011             & \multirow{3}{*}{0.017} \\
\cline{1-2} \cline{4-4}
3                        & 41.076          &                          & 0.021             &                          \\
\cline{1-2} \cline{4-4}
5                        & 45.923           &                          & 0.030             &                          \\
\hline
\end{tabular}
\caption{Performance and computation time of MVE-DDPG with varying ensemble sizes on the \textit{Swimmer} environment, averaged over 5 random seeds. An ensemble of 5 critic networks enables MVE-DDPG to achieve performance comparable to HAC.}
\label{tab:EnSwimmer}
\end{table}

Table ~\ref{tab:EnSwimmer} present the results of ablation studies conducted on \textit{Swimmer} environment. These results demonstrate that HAC achieves competitive performance when compared to MVE-DDPG with ensemble critic sizes of 3 and 5. Although HAC requires more computation per gradient step due to the backward costate calculation, it offers overall computational efficiency. Specifically, while MVE-DDPG relies on multiple critic networks to stabilize learning, HAC avoids this overhead, thus reducing total training time. This advantage becomes increasingly significant in complex, high-dimensional environments, where accurately learning the critic distribution necessitates larger ensemble sizes for MVE-DDPG.


\end{document}